\documentclass[10pt,twocolumn,twoside]{IEEEtran}

\usepackage[utf8]{inputenc}
\usepackage{amsmath,amsthm,amsfonts,amssymb,bm,accents}
\usepackage{graphicx,cite,enumerate,afterpage}
\graphicspath{{Images/}}

\usepackage{color}
\definecolor{pennblue}{cmyk}{1,0.65,0,0.30}
\definecolor{pennred}{cmyk}{0,1,0.65,0.34}
\definecolor{mygreen}{rgb}{0.10,0.50,0.10}

\usepackage{blindtext}

\newcommand{\includesvg}[2][scale=1]{\includegraphics[#1]{#2.pdf}}

\usepackage[english]{babel}

\usepackage{algpseudocode,algorithm}
\algrenewcommand\algorithmicdo{}
\algrenewtext{EndFor}{\algorithmicend}
\algrenewtext{EndProcedure}{\algorithmicend}

\newtheorem{theorem}{Theorem}
\newtheorem*{theorem*}{Theorem}
\newtheorem{proposition}{Proposition}
\newtheorem{corollary}{Corollary}
\newtheorem{lemma}{Lemma}
\theoremstyle{definition}

\theoremstyle{definition}
\newtheorem{remark}{Remark}
\newtheorem*{remark*}{Remark}
\newtheoremstyle{assume}
  {3pt}
  {3pt}
  {}
  {}
  {\bf}
  {}
  { }
  {\thmname{#1}.\thmnumber{#2}\thmnote{ \textnormal{(\textit{#3})}}}
\theoremstyle{assume}

\DeclareMathOperator{\E}{\mathbb{E}}

\DeclareMathOperator{\trace}{Tr}
\DeclareMathOperator{\diag}{diag}

\DeclareMathOperator*{\argmin}{argmin}
\DeclareMathOperator*{\argmax}{argmax}
\DeclareMathOperator*{\colspan}{colspan}

\newcommand{\vect}[2]{\ensuremath{[\begin{array}{#1} #2 \end{array}]}}
\newcommand{\norm}[1]{\ensuremath{\left\| #1 \right\|}}
\newcommand{\abs}[1]{\ensuremath{{\left\vert #1 \right\vert}}}

\newcommand{\calA}{\ensuremath{\mathcal{A}}}
\newcommand{\calB}{\ensuremath{\mathcal{B}}}

\newcommand{\calG}{\ensuremath{\mathcal{G}}}

\newcommand{\calK}{\ensuremath{\mathcal{K}}}

\newcommand{\calO}{\ensuremath{\mathcal{O}}}

\newcommand{\calS}{\ensuremath{\mathcal{S}}}
\newcommand{\calT}{\ensuremath{\mathcal{T}}}

\newcommand{\calV}{\ensuremath{\mathcal{V}}}

\newcommand{\calX}{\ensuremath{\mathcal{X}}}

\newcommand{\bzero}{\ensuremath{\bm{0}}}
\newcommand{\bA}{\ensuremath{\bm{A}}}
\newcommand{\bB}{\ensuremath{\bm{B}}}
\newcommand{\bC}{\ensuremath{\bm{C}}}
\newcommand{\bD}{\ensuremath{\bm{D}}}

\newcommand{\bH}{\ensuremath{\bm{H}}}
\newcommand{\bI}{\ensuremath{\bm{I}}}

\newcommand{\bK}{\ensuremath{\bm{K}}}
\newcommand{\bL}{\ensuremath{\bm{L}}}

\newcommand{\bP}{\ensuremath{\bm{P}}}

\newcommand{\bR}{\ensuremath{\bm{R}}}

\newcommand{\bV}{\ensuremath{\bm{V}}}
\newcommand{\bW}{\ensuremath{\bm{W}}}
\newcommand{\bX}{\ensuremath{\bm{X}}}
\newcommand{\bY}{\ensuremath{\bm{Y}}}
\newcommand{\bZ}{\ensuremath{\bm{Z}}}

\newcommand{\bb}{\ensuremath{\bm{b}}}

\newcommand{\br}{\ensuremath{\bm{r}}}
\newcommand{\bs}{\ensuremath{\bm{s}}}

\newcommand{\bu}{\ensuremath{\bm{u}}}
\newcommand{\bv}{\ensuremath{\bm{v}}}
\newcommand{\bw}{\ensuremath{\bm{w}}}
\newcommand{\bx}{\ensuremath{\bm{x}}}
\newcommand{\by}{\ensuremath{\bm{y}}}
\newcommand{\bz}{\ensuremath{\bm{z}}}
\newcommand{\bLambda}{\ensuremath{\bm{\Lambda}}}
\newcommand{\bSigma}{\ensuremath{\bm{\Sigma}}}

\newcommand{\setR}{\ensuremath{\mathbb{R}}}
\newcommand{\setS}{\ensuremath{\mathbb{S}}}
\newcommand{\setC}{\ensuremath{\mathbb{C}}}

\def\st/{\textsuperscript{st}}
\def\nd/{\textsuperscript{nd}}
\def\rd/{\textsuperscript{rd}}
\def\th/{\textsuperscript{th}}

\newcommand{\del}{\ensuremath{\partial}}

\newcommand{\bxb}{\ensuremath{\bar{\bx}}}
\newcommand{\byb}{\ensuremath{\bar{\by}}}

\newcommand{\bxh}{\ensuremath{\hat{\bx}}}
\newcommand{\bzh}{\ensuremath{\hat{\bz}}}

\newcommand{\byt}{\ensuremath{\tilde{\by}}}

\newcommand{\MSE}{\textup{MSE}}
\newcommand{\bKb}{\ensuremath{\bm{{\bar{K}}}}}

\newcommand{\nodeset}{\calV}

\newcommand{\graph}{\mathbb{G}}
\newcommand{\levb}[1]{\bar{\ell}_{#1}}

\title{Greedy Sampling of Graph Signals}

\author{Luiz~F.~O.~Chamon and~Alejandro~Ribeiro%
\thanks{Department of Electrical and Systems Engineering, University of Pennsylvania.
 e-mail: \mbox{\texttt{luizf@seas.upenn.edu}}, \mbox{\texttt{aribeiro@seas.upenn.edu}}.
This work was supported by NSF CCF 1717120 and ARO W911NF1710438.
Part of the results in this paper appeared in~[24] and~[32].}%
}

\begin{document}
\maketitle
\begin{abstract}
Sampling is a fundamental topic in graph signal processing, having found applications in estimation, clustering, and video compression. In contrast to traditional signal processing, the irregularity of the signal domain makes selecting a sampling set non-trivial and hard to analyze. Indeed, though conditions for graph signal interpolation from noiseless samples exist, they do not lead to a unique sampling set. The presence of noise makes choosing among these sampling sets a hard combinatorial problem. Although greedy sampling schemes are commonly used in practice, they have no performance guarantee. This work takes a twofold approach to address this issue. First, universal performance bounds are derived for the Bayesian estimation of graph signals from noisy samples. In contrast to currently available bounds, they are not restricted to specific sampling schemes and hold for any sampling sets. Second, this paper provides near-optimal guarantees for greedy sampling by introducing the concept of approximate submodularity and updating the classical greedy bound. It then provides explicit bounds on the approximate supermodularity of the interpolation mean-square error showing that it can be optimized with worst-case guarantees using greedy search even though it is not supermodular. Simulations illustrate the derived bound for different graph models and show an application of graph signal sampling to reduce the complexity of kernel principal component analysis.
\end{abstract}

\section{Introduction}
	\label{S:Intro}

Graph signal processing~(GSP) is an emerging field that studies signals supported on irregular domains~\cite{Shuman13e, Sandryhaila13d}. It extends traditional signal processing techniques to more intricate data structures, finding applications in sensor networks, image processing, and clustering, to name a few~\cite{Narang12p, Tremblay16c, Zhu12a}. Extensions of sampling, in particular, have attracted considerable interest from the GSP community~\cite{Pesenson08s, Pesenson10s, Narang13s, Shomorony14s, Chen15d, Anis16e, Tsitsvero16s, Chen16s, Antonio16s, Chepuri16s}. This is not surprising given the fundamental role of sampling in signal processing. Sampling methods in GSP are broadly divided into two categories: \emph{selection sampling}, in which the graph signal is observed at a subset of nodes~\cite{Chen15d}, and \emph{aggregation sampling}, in which the signal is observed at a single node for many applications of the graph shift~\cite{Antonio16s}. This work focuses on the former.

As in classical signal processing, samples are only useful inasmuch as they represent the original signal. Conditions under which it is possible to reconstruct a graph signal from noiseless samples can be found in~\cite{Pesenson08s, Pesenson10s, Narang13s, Shomorony14s, Chen15d, Anis16e, Tsitsvero16s}. These, however, do not necessarily lead to unique sampling sets. In fact, for finite graphs with bounded weights, it can be shown that almost every sampling set larger than the bandwidth of the signal guarantees perfect reconstruction~\cite{Shomorony14s, Chen15d, Anis16e, Tsitsvero16s}. In the presence of noise, however, it is not straightforward which of these exponentially many sets performs best, an issue that becomes more severe as the measurement signal-to-noise ratio~(SNR) decreases. In general, selecting an optimal sampling set is NP-hard~\cite{Krause08n, Das11s, Sagnol13a, Ranieri14n}.

In~\cite{Chen15d, Chen16s}, this issue was addressed using randomized sampling schemes, for which optimal sampling distributions and performance bounds were derived for different types of graphs and graph signals. For high SNR, it was shown that sampling proportionally to the leverage score~(or its square-root) approximates the sampling distribution that minimizes the reconstruction mean-square error~(MSE). Alternatively, a convex relaxation approach was adopted in~\cite{Gama16r}, where the sampling set selection problem was cast as a binary semi-definite program~(SDP) and solved by relaxing the binary constraint and thresholding the solution. Rounding and truncation can also be used to approximate the solution of this binary problem. Nevertheless, greedy sampling remains pervasive and has proven successful in many applications~\cite{Chen15d, Anis16e, Shomorony14s, Thanou14l, Tsitsvero16s, Gama16r}, though performance analyses are available only for surrogate figures of merit of the MSE, such as the log-determinant~\cite{Chepuri16s}.

To be sure, this success is warranted by the attractive features of greedy algorithms for large-scale problems. First, their complexity is polynomial in the deterministic case and randomized versions exist that are linear in the size of the ground set, which in this case is the number of nodes in the graph~\cite{Mirzasoleiman15l}. Also, since they build the solution sequentially, they can be interrupted at any time if, for instance, a desired performance level is reached. More importantly, there is an upper bound on the suboptimality of the greedy solution to monotonic supermodular function minimization problems. This is indeed why greedy algorithms are often used in sensor selection, experimental design, and machine learning~\cite{Ranieri14n, Krause08n, Sagnol13a, Das11s, Bach14l}. However, the main performance measure in GSP, namely the MSE, is not supermodular in general~\cite{Chamon16n}.

In this work, we study the reconstruction~(interpolation) performance of greedy sampling schemes in GSP and set out to reconcile the empirical success of greedy MSE minimization with the fact that it is not supermodular. First, in contrast to~\cite{Chen15d, Chen16s}, we adopt a Bayesian approach to graph signal estimation and consider sampling to be deterministic~(Section~\ref{S:GraphSignal}). Then, we derive bounds on the interpolation MSE that are universal in the sense that they hold for all sampling sets and any sampling method~(Section~\ref{S:Bounds}). These universal bounds are explicit, tractable, and provide practical means of benchmarking the MSE performance of any sampling scheme. Numerical analyses show that the bounds are tight when signal and noise are homeoscedastic. Finally, we develop the concept of \emph{approximate supermodularity} introduced in~\cite{Chamon16n} and provide near-optimal guarantees for the greedy minimization of the interpolation MSE~(Section~\ref{S:greedySampling}). This result justifies the use of greedy sampling set selection in GSP and explains its success.

To illustrate the practical value of these results, we recall that the concept of sampling is also at the core of statistical methods, such as data subsetting and variable selection, that are crucial for \emph{big data} applications~\cite{Woodruff14s, Feldman13t}. Kernel methods, in particular, are prone to complexity issues in large data sets. For instance, performing kernel principal component analysis~(kPCA) on a data set of size~$n$ requires $n^2$~kernel evaluations~(KEs) and $\Theta(n^3)$~operations, while extracting projections for new data takes $n$~KEs and $\Theta(np)$~operations, where $p$ is the number of principal components~(PCs) retained~\cite{Scholkopf98n, Jeronimo13k}. We show that this problem can be cast in the context of GSP and that greedy sampling can be used to reduce the complexity of projections by over~$95\%$ at a small performance cost~(Section~\ref{S:kPCA}).

\textbf{Notation}: Lowercase boldface letters represent vectors~($\bx$), uppercase boldface letters are matrices~($\bX$), and calligraphic letters denote sets~($\calA$). We write $\abs{\calA}$ for the cardinality of~$\calA$ and denote the empty set by~$\{\}$. Set subscripts refer either to the vector obtained by keeping only the elements with indices in the set~($\bx_\calA$) or to the submatrix whose columns have indices in the set~($\bX_\calA$). To say~$\bX$ is a positive semi-definite~(PSD) matrix we write~$\bX \succeq 0$, so that for $\bX,\bY \in \setC^{n \times n}$, $\bX \preceq \bY \Leftrightarrow \bb^{H} \bX \bb \leq \bb^{H} \bY \bb$, for all $\bb \in \setC^n$. The set of PSD matrices is denoted~$\setS_{+}$ and the set of non-negative real numbers is denote~$\setR_{+}$. Finally, we take the derivative of a function $f$ with respect to an $n \times 1$ vector $\bx$ to yield the $1 \times n$ gradient vector, i.e.,
$\del f/\del \bx = \left[\ \del f/\del x_1 \ \cdots \ \del f/\del x_n \ \right]$~\cite{Kailath00l}.

\section{Sampling and Interpolation of Graph Signals}
	\label{S:GraphSignal}

A graph-supported signal, or \emph{graph signal} for short, is an assignment of values to the nodes of a graph. Formally, let $\graph$ be a weighted graph with node set $\nodeset$, having cardinality $\abs{\nodeset} = n$, and define a graph signal to be an injective mapping $\sigma: \nodeset \to \setC$. For an ordering of the nodes in~$\calV$, this signal can be represented as an $n \times 1$ vector that captures its values at each node:
\begin{equation}\label{E:graphSignalVector}
	\bx = \vect{ccc}{ \sigma(u_1) & \cdots & \sigma(u_n) }^T
		\text{,} \quad u_i \in \nodeset
		\text{.}
\end{equation}
In what follows, we assume that the node ordering is fixed, so that we can index $\bx$ using elements of~$\calV$. For instance, we write~$\bx_{\{u_i, u_j, u_k\}} = \left[\ \sigma(u_i) \ \sigma(u_j) \ \sigma(u_k) \ \right]^T$.

Of interest to GSP is the spectral representation of the signal~$\sigma$~(or~$\bx$), which depends on the graph on which it is supported. Indeed, let~$\bA \in \setC^{n \times n}$ be a matrix representation of~$\graph$. Usual choices include the adjacency matrix or one of the discrete Laplacians~\cite{Shuman13e, Sandryhaila13d}. Assume that $\bA$ is consistent with the signal vector~\eqref{E:graphSignalVector} in the sense that they employ the same ordering of the nodes in $\nodeset$. Furthermore, assume that $\bA$ is normal, i.e., that there exist~$\bV \in \setC^{n \times n}$ unitary and~$\bD \in \setR^{n \times n}$ diagonal such that~$\bA = \bV \bD \bV^{H}$, where $\cdot^{H}$ is the Hermitian~(conjugate transpose) operator~\cite{Horn13}. Then, the \emph{graph Fourier transform} of~$\bx$ is given by~\cite{Shuman13e, Sandryhaila13d}
\begin{equation}\label{E:GFT}
	\bxb = \bV^{H} \bx
		\text{.}
\end{equation}
Observe that if $\bA$ is normal we obtain a spectral energy conservation property analog to Parseval's theorem in classical signal processing: it is ready to see that~$\norm{\bxb}_2 = \norm{\bx}_2$ if and only if~$\bV$ in~\eqref{E:GFT} is unitary, which holds if and only if~$\bA$ is normal~\cite{Horn13}.

Similar to traditional signal processing, a graph signal~$\bx$ is said to be \emph{spectrally sparse}~(\emph{ssparse}) when its spectral representation is sparse. Explicitly, $\bx$ is~\emph{$\calK$-ssparse} if~$\bxb$ in~\eqref{E:GFT} is such that~$\bxb_{\nodeset \setminus \calK}$ is a zero vector. Then,
\begin{equation}\label{E:bandlimited}
	\bx = \bV_\calK \bxb_\calK
		\text{.}
\end{equation}
Note that spectrally sparse signals are a superset of bandlimited~(``low-pass'') signals. Hence, all results in this work apply to bandlimited signals regardless of the graph frequency order adopted~\cite{Chen16s, Anis16e, Chen15d}.

The interest in $\calK$-ssparse or bandlimited graph signals is motivated similarly to traditional signal processing: these signals can be sampled and interpolated without loss of information. Indeed, take sampling to be the operation of observing the value of a graph signal on $\calS \subseteq \nodeset$, the \emph{sampling set}. Then, there exists a set~$\calS$ of size~$\abs{\calK}$ such that $\bx$ can be recovered exactly from $\bx_\calS$~\cite{Shomorony14s, Chen15d, Anis16e, Tsitsvero16s}. If, however, only a corrupted version of~$\bx_\calS$ is available, then~$\bx$ can only be approximated. To do so, the next section poses noisy interpolation as a Bayesian estimation problem, from which the minimum MSE interpolation operator can be derived. This then allows us to provide universal bounds on the reconstruction error and give near-optimal guarantees for greedy sampling strategies.

\subsection{Graph signal interpolation}
	\label{S:Interpolation}

We study graph signal interpolation as a Bayesian estimation problem. Formally, let~$\bx \in \setC^n$ be a graph signal and~$\calS \subseteq \calV$ be a sampling set. We wish to estimate
\begin{equation}\label{E:z}
	\bz = \bH \bx
		\text{,}
\end{equation}
for some matrix~$\bH \in \setC^{m \times n}$ based on the samples~$\by_\calS$ taken from
\begin{equation}\label{E:y}
	\by = \bx + \bw
		\text{,}
\end{equation}
where~$\bw \in \setC^n$ is a circular zero-mean noise vector. By circular we mean that its \emph{relation matrix} vanishes, i.e., that~$\E \bw \bw^T = \bzero$~\cite{Adali14o}. Note that~\eqref{E:z} accounts for scenarios in which we are not interested in the graph signal itself but on a post-processed value, such as the output of a linear classifier or estimator~(e.g., Section~\ref{S:kPCA}). The usual graph signal interpolation problem from~\cite{Chamon16n, Chamon16u, Shomorony14s, Chen15d, Anis16e, Tsitsvero16s, Chen16s} is recovered by taking~$\bH = \bI$.

The prior distribution of~$\bx$ reflects the fact that the graph signal is~$\calK$-ssparse by assuming it is a circular zero-mean distribution with covariance matrix~$\bSigma = \bx \bx^H = \bV_\calK \bLambda \bV_\calK^H$ for~$\bLambda = \diag(\lambda_i)$, $\lambda_i \in \setR_+$. We assume without loss of generality that~$\bLambda$ is full-rank. Otherwise, remove from~$\calK$ any element~$i$ for which~$\lambda_i = 0$. Note that this is equivalent to placing a zero-mean uncorrelated prior on~$\bxb$ in~\eqref{E:GFT}. Hence, this model can also be interpreted as the generative model for a \emph{wide-sense stationary} random process on~$\graph$~\cite{Girault15s, Antonio16st, Perraudin16s}. The noise prior is taken as a zero-mean circular distribution with covariance matrix~$\bLambda_w = \diag(\lambda_{w,i})$, $\lambda_{w,i} \in \setR_{+}$ and~$\lambda_{w,i} > 0$.

We consider estimates of~$\bz$ of the form
\begin{equation}\label{E:zhat}
	\bzh(\calS) = \bL(\calS) \by_\calS
		\text{,}
\end{equation}
for some~$\bL(\calS) \in \setC^{n \times \abs{\calS}}$. Because~$\bL$ recovers~(approximates) $\bz$ from the samples~$\by_\calS$, it is referred to as a \emph{linear interpolation operator}~\cite{Chen15d, Chen16s, Anis16e}. An optimal interpolation operator can be found for each~$\calS$ by minimizing the interpolation error covariance matrix as in
\begin{equation}\label{P:interpOP}
\begin{aligned}
	&\underset{\bL}{\text{minimize}} &&\bK \left[ \bzh(\calS) \right]
	\\
	&\text{subject to} &&\bzh(\calS) = \bL \by_\calS
\end{aligned}
\end{equation}
where~$\bK\left[ \bzh(\calS) \right] = \E \left[ \left( \bz - \bzh(\calS) \right) \left( \bz - \bzh(\calS) \right)^{H} \mid \bx,\bw \right]$ and the minimum is taken with respect to the partial ordering of the PSD cone~(see Remark~\ref{R:psd_minimum}). We have omitted the dependence of~$\bL$ on~$\calS$ for clarity. Our interest in solving~\eqref{P:interpOP} instead of minimizing the MSE directly is that it is more general. In particular, a solution of~\eqref{P:interpOP} also minimizes any spectral function of~$\bK$, including the MSE and the log-determinant. The following proposition gives an explicit solution to this problem. To clarify the derivations, we define the selection matrix~$\bC \in \{0,1\}^{\abs{\calS} \times N}$ composed of the identity matrix rows with indices in~$\calS$, so that the samples of~\eqref{E:y} can be written as~$\by_\calS = \bC \by$.

\begin{proposition}\label{T:optimalInterpolator}

Let~$\by = \bx + \bw$ be noisy observations of a graph signal~$\bx$. Let the priors on~$\bx$ and~$\bw$ be zero-mean circular distributions with covariances~$\bSigma  = \bV_\calK \bLambda \bV_\calK^{H}$, $\bLambda = \diag(\lambda_i)$, and~$\bLambda_w = \diag(\lambda_{w,i})$ respectively. Given a sampling set~$\calS$, the optimal Bayesian linear interpolator~$\bL^\star$ that solves problem~\eqref{P:interpOP} is obtained as a solution of
\begin{equation}\label{E:normalEq}
	\bL^\star \bC \left( \bSigma + \bLambda_w \right) \bC^{T}
		= \bH \bSigma \bC^{T}
		\text{.}
\end{equation}
The error covariance matrix of the optimal interpolation~$\bzh^\star = \bL^\star \by_\calS$ is given by
\begin{equation}\label{E:Kstar}
	\bK^\star(\calS) = \bH \bV_\calK \left( \bLambda^{-1} +
		\sum_{i \in \calS} \lambda_{w,i}^{-1} \bv_i \bv_i^{H}  \right)^{-1} \bV_\calK^{H} \bH^{H}
		\text{,}
\end{equation}
where~$\bV_\calK = \vect{ccc}{\bv_1 & \cdots & \bv_N}^H$.

\end{proposition}

\begin{proof}

Start by substituting~\eqref{E:z} and~\eqref{E:zhat} into the definition of~$\bK$ to get
\begin{equation*}
	\bK(\bL \by_\calS) = \E \left[ (\bH \bx - \bL \bC \by)
		(\bH \bx - \bL \bC \by)^{H}
			\:\middle\vert\: \bx,\bw \right]
		\text{.}
\end{equation*}
Note that we used the fact that~$\by_\calS = \bC \by$. Then, using the priors on~$\bx$ and~$\bw$, $\bK$ expands to
\begin{equation}\label{E:KLS}
\begin{aligned}
	\bK(\bL \by_\calS) &= \bH \bSigma \bH^{H}
		- \bL \bC \bSigma \bH^{H}
	- \bH \bSigma \bC^{T} \bL^{H}
	\\
	{}&+ \bL \bC (\bSigma + \bLambda_w) \bC^{T} \bL^{H}
		\text{.}
\end{aligned}
\end{equation}

From the partial ordering of the PSD cone, $\bL^\star$ can be obtained by minimizing the scalar cost function
\begin{equation}\label{E:J}
	J(\bL) = \bb^{H} \bK(\bL \by_\calS) \bb
\end{equation}
\emph{simultaneously} for all $\bb \in \setC^n$~\cite{Kailath00l}. Substituting~\eqref{E:KLS} into~\eqref{E:J} and setting its gradient with respect to $\bb^H \bL$ to zero gives
\begin{equation*}
	\frac{\del J(\bL)}{\del \bb^{H} \bL} = \bzero \Leftrightarrow
	\bC \left( \bSigma + \bLambda_w \right) \bC^{T} \bL^{H} \bb
		= \bC \bSigma \bH^{H} \bb
		\text{.}
\end{equation*}
%
Since this must hold for all~$\bb$ simultaneously, we obtain~\eqref{E:normalEq}.

To determine the error covariance matrix~$\bK^\star$ of the optimal interpolator, replace any~$\bL^\star$ satisfying~\eqref{E:normalEq} into~\eqref{E:KLS} and expand~$\bSigma = \bV_\calK \bLambda \bV_\calK^{H}$ to get
\begin{multline}\label{E:Kstar1}
	\bK^\star(\bL^\star \by_\calS) = \bH \bV_\calK \left\{
		\bLambda - \bLambda \bV_\calK^{H} \bC^{T} \times{}
	\vphantom{\sum}\right.
	\\
	\left.\vphantom{\sum}
		\left[ \bC \left( \bV_\calK \bLambda \bV_\calK^{H}
			+ \bLambda_w \right) \bC^{T} \right]^{-1}
			\bC \bV_\calK \bLambda
	\right\} \bV_\calK^{H} \bH^{H}
		\text{.}
\end{multline}
Note that~\eqref{E:Kstar1} does not depend on $\bL^\star$ or $\by_\calS$, only on the sampling set~$\calS$ through the selection matrix~$\bC$. Moreover, since $\bLambda_w$ is diagonal and full rank, $(\bC \bLambda_w \bC^{T})^{-1} = \bC \bLambda_w^{-1} \bC^{T}$, so that the inverse in~\eqref{E:Kstar1} always exists. Therefore, using the matrix inversion lemma~\cite{Horn13} gives
\begin{equation*}
	\bK^\star(\calS) = \bH \bV_\calK \left( \bLambda^{-1} +
		\bV_\calK^{H} \bC^{T} \bC \bLambda_w^{-1} \bC^T \bC \bV_\calK
		\right)^{-1} \bV_\calK^{H} \bH^{H}
		\text{.}
\end{equation*}
Given that~$\bC^{T} \bC$ is a diagonal matrix with ones on the indices in~$\calS$ and zeros everywhere else, we obtain~\eqref{E:Kstar} by noting that
\begin{equation*}
	\bV_\calK^{H} \bC^{T} \bC \bLambda_w^{-1} \bC^{T} \bC \bV_\calK =
	\sum_{i \in \calS} \lambda_{w,i}^{-1} \bv_i \bv_i^{H}
		\text{,}
\end{equation*}
for~$\bV_\calK = \vect{ccc}{\bv_1 & \cdots & \bv_N}^H$. \qedhere

\end{proof}

%
%

Given prior distributions for the graph signal and noise, Proposition~\ref{T:optimalInterpolator} determines the optimal linear interpolator from the samples in~$\calS$. If the priors on~$\bx$ and~$\bw$ are moreover Gaussian, then~$\bzh^\star = \bL^\star \by_\calS$ is also the maximum likelihood estimate of~$\bz$~\cite{Kailath00l}. An important consequence of the Bayesian statement of Proposition~\ref{T:optimalInterpolator} is that~$\bLambda$ and~$\bLambda_w$ are taken from prior distributions on the signal and noise. Thus, their actual values need not be known exactly, as illustrated in Section~\ref{S:kPCA}. Note that the optimal error covariance matrix~$\bK^\star$ now depends only on the sampling set~$\calS$, since it measures the error of the optimal estimator~$\bL^\star$. Moreover, although we assume that the interpolation is performed as a single step projection, iterative procedures can also be used~\cite{Lorenzo16a, Wang15l}.

Despite our assumption that~$\bLambda_w$ is full-rank, \eqref{E:normalEq} also holds in the noiseless case~($\bLambda_w = \bzero$). Its solution, however, may no longer be unique. In particular, this happens if the sampling set is not sufficient to determine~$\bz$, i.e., if~$\bC \bV_\calK$ is rank-deficient~\cite{Shomorony14s, Chen15d, Anis16e, Tsitsvero16s}. In contrast, when~$\bLambda_w \succ \bzero$, the matrix on the left-hand side of~\eqref{E:normalEq} is always invertible and~$\bL^\star$ is unique for each~$\calS$. This is similar to the well-known regularization effect of noise in Kalman filtering~\cite{Kailath00l}. The interpolation performance given in~\eqref{E:Kstar}, however, is not the same for all sampling sets.

\begin{remark}\label{R:psd_minimum}

Problem~\eqref{P:interpOP} is a PSD matrix minimization problem that searches for the optimal interpolator~$\bL^\star$ that minimizes the error covariance matrix~$\bK$. In general, optimization problems in the PSD cone need not have a solution. Since the ordering of PSD matrices is only partial, the existence of a matrix that is smaller than all other matrices is not guaranteed~\cite{Boyd04c}. As shown in Proposition~\ref{T:optimalInterpolator}, this is not the case here. Problem~\eqref{P:interpOP} admits a dominant solution~$\bL^\star$ in the PSD cone, i.e., it holds that~$\bK(\bL^\star \by_\calS) \preceq \bK(\bL \by_\calS)$ for all~$\bL \in \setC^{n \times \abs{\calS}}$. This means that~$\bL^\star$ minimizes all the eigenvalues of~$\bK$ simultaneously. Equivalently, it implies that~$\bL^\star$ is a solution to the minimization of any spectral function of~$\bK$. In particular, it follows that $\bL^\star$ minimizes the MSE, since~$\MSE(\bzh) := \E \norm{\bz - \bzh}_2^2 = \trace \left[ \bK(\bzh) \right]$, and the~$\log\det\left[ \bK(\bzh) \right]$.

\end{remark}

\subsection{Sampling set selection}
	\label{S:SSS}

Proposition~\ref{T:optimalInterpolator} allows us to evaluate the optimal interpolator~$\bL^\star$ that minimizes the estimation error covariance matrix for a given sampling set. This does not guarantee, however, that there is no other sampling set of the same size for which the interpolation error is smaller. To address this issue, we investigate the \emph{sampling set selection} problem which sets out to find the sampling set that minimizes the interpolation MSE over all sampling sets. Explicitly, we wish to solve
\begin{equation}\label{P:SSS}
\begin{aligned}
	&\underset{\calS \subseteq \calV}{\text{minimize}} && \MSE(\calS)
	\\
	&\text{subject to} &&\abs{\calS} \leq k
\end{aligned}
\end{equation}
where $\MSE(\calS) = \trace[\bK^\star(\calS)]$.

An important fact about~\eqref{P:SSS} is that increasing~$\calS$ always decreases MSE. This has two important consequences. First, the unconstrained version of~\eqref{P:SSS} is trivial, i.e., its solution is~$\calS = \calV$. Second, it implies that the constraint in~\eqref{P:SSS} is tight, i.e., it can be replaced by the equality constraint~$\abs{\calS} = k$ without changing the problem solution. This property is a direct corollary of the following lemma and the monotonicity of the trace operator~\cite{Bhatia97m}:

\begin{lemma}\label{T:monotonicity}

The matrix-valued set function $\bK^\star(\calS)$ in~\eqref{E:Kstar} is monotonically decreasing with respect to the PSD cone, i.e., $\bK^\star(\calA) \succeq \bK^\star(\calB)$ whenever~$\calA \subseteq \calB \subseteq \calV$.

\proof

Start by noting that $\bK^\star$ in~\eqref{E:Kstar} can be written as
\begin{equation*}
	\bK^\star(\calS) = \bH \bV_\calK \bKb(\calS) \bV_\calK^{H} \bH^{H}
		\text{,}
\end{equation*}
with $\bKb(\calS) = \left[ \bLambda^{-1} + \bR(\calS) \right]^{-1}$ and
\begin{equation*}
	\bR(\calS) = \sum_{i \in \calS} \lambda_{w,i}^{-1} \bv_i \bv_i^{H}
		\text{.}
\end{equation*}
Since $\bK^\star$ and $\bKb$ are congruent, it suffices to show that $\bKb$ is a monotonically decreasing set function~\cite{Horn13}.

To do so, note that $\bKb$ only depends on $\calS$ through $\bR(\calS)$ and that $\bR$ is additive, i.e., $\bR(\calA \cup \calB) = \bR(\calA) + \bR(\calB)$. Then, since $\lambda_{w,i} > 0$, $\bR$ is a sum of PSD matrices, which implies that~$\calA \subseteq \calB \Rightarrow \bR(\calA) \preceq \bR(\calB)$, i.e., $\bR$ is monotonically increasing. From the antitonicity of the matrix inverse~\cite{Bhatia97m}, it follows that~$\bKb$ is monotonically decreasing.
\qed

\end{lemma}

Although Lemma~\ref{T:monotonicity} reduces the searching space to sampling sets of size~$k$, \eqref{P:SSS} remains a combinatorial optimization problem: $\binom{n}{k}$ sampling sets must still be checked, which is impractical even for moderately small~$n$. In fact, due to the irregularity of the domain of graph signals, sampling set selection is NP-hard in general. It is straightforward to see that it is equivalent to the sensor placement or forward regression problems in~\cite{Krause08n, Das11s, Sagnol13a, Ranieri14n}, so that the typical reduction from set cover applies~\cite{Natarajan95s}.

In the following sections, we address this issue in two ways. First, we derive universal performance bounds that hold for all sampling sets~(Section~\ref{S:Bounds}). These bounds can therefore be used to evaluate the quality of a sampling set or selection heuristic \emph{a posteriori}. Second, we study the greedy sampling algorithm and provide near-optimal \emph{a priori} guarantees based on the concept of \emph{approximate submodularity}~(Section~\ref{S:greedySampling}). Special cases of these results that considered real-valued homeoscedastic signal and noise~($\bLambda = \sigma_x^2 \bI$ and $\bLambda_w = \sigma_w^2 \bI$) and no transformation of the graph signal~($\bH = \bI$) can be found in~\cite{Chamon16u, Chamon16n}.

\section{Universal Bounds on the Interpolation MSE}
	\label{S:Bounds}

In this section, we derive interpolation performance bounds that hold for all~$\calS$. These universal bounds can be used to inform the sampling set selection by (i)~describing how different factors influence the reconstruction performance and (ii)~gauging the quality of sampling set instances. The main result of this section is presented below.

\begin{theorem}\label{T:mseBound}

Let $\bx$ be a $\calK$-ssparse stationary graph signal and $\by = \bx + \bw$ be its noisy observations. Take~$\bzh^\star = \bL^\star \by_\calS$ to be the minimum MSE linear interpolation of~$\bz = \bH \bx$ based on a sampling set~$\calS$ given zero-mean circular priors on the signal and noise such that~$\E \bx\bx^H = \bV_\calK \bLambda \bV_\calK^H$ and~$\E \bw \bw^H = \bLambda_w$. For~$\bW = \bV_\calK^H \bH^H \bH \bV_\calK \succ 0$, the reconstruction error $\MSE(\calS) = \E \norm{\bz - \bzh^\star}^2 = \trace[\bK^\star(\calS)]$ is bounded by
	\begin{equation}\label{E:mseBound}
		\frac{\abs{\calK}^2}{\trace \left[ (\bW \bLambda)^{-1} \right] + \levb{\abs{\calS}}}
			\leq
		\MSE(\calS)
			\leq
		\trace \left( \bW \bLambda \right)
			\text{,}
	\end{equation}
	where~$\levb{m}$ is the sum of the $m$~largest weighted structural SNRs~$\ell_i = \lambda_{w,i}^{-1} \norm{\bv_i}_{\bW^{-1}}^2$, with~$\bV_\calK = \vect{ccc}{\bv_1 & \cdots & \bv_N}^H$ and~$\norm{\bx}_{\bA}^2 = \bx^{H} \bA \bx$. Explicitly, $\levb{m} = \max_{\calX:\abs{\calX} = m} \sum_{j \in \calX} \ell_j$.

\end{theorem}

\begin{proof}
	Start with the upper bound that is achieved for an empty sampling set, i.e., for $\calS = \{\}$. Indeed, recall from Lemma~\ref{T:monotonicity} that $\bK^\star$ is a monotone decreasing set function, i.e., it achieves its maximum for the empty set. Thus, it holds that~$\bK(\bxh^\star) \preceq \bH \bV_\calK \bLambda \bV_\calK^{H} \bH^{H}$, from which the upper bound in~\eqref{E:mseBound} follows by the monotonicity of the trace operator~\cite{Bhatia97m}.

	To obtain the lower bound, start by using~\eqref{E:Kstar} to get
	\begin{equation}\label{E:mseBound1}
		\MSE(\calS) = \trace\left[
			\bW \left( \bLambda^{-1} + \sum_{i \in \calS} \lambda_{w,i}^{-1} \bv_i \bv_i^{H} \right)^{-1}
		\right]
			\text{,}
	\end{equation}
	where we used the circular commutation property of the trace. Then, since the trace of a matrix is the sum of its eigenvalues, the arithmetic/harmonic means inequality can be used to get, for any~$N \times N$ positive-definite matrix~$\bX$,
	\begin{equation*}
		\trace\left( \bX \right) \geq \frac{N^2}{\trace\left( \bX^{-1} \right)}
			\text{,}
	\end{equation*}
	with equality if and only if $\bX = \gamma \bI$, $\gamma > 0$~\cite{Horn13}. Since $\bW \succ 0$, the matrix in~\eqref{E:mseBound1} is positive-definite and we have
	\begin{equation*}
		\MSE(\calS) \geq
			\frac{
				\abs{\calK}^2
			}{
				\trace \left[ (\bW \bLambda)^{-1} \right] +
				\trace \left[ \bW^{-1} \left( \sum_{i \in \calS} \lambda_{w,i}^{-1} \bv_i \bv_i^{H} \right) \right]
			}
			\text{,}
	\end{equation*}
	which from the commutation property of the trace gives
	\begin{equation*}
		\MSE(\calS) \geq
			\frac{\abs{\calK}^2}{\trace \left[ (\bW \bLambda)^{-1} \right] +
				\sum_{i \in \calS} \lambda_{w,i}^{-1} \norm{\bv_i}_{\bW^{-1}}^2}
			\text{,}
	\end{equation*}
	where~$\norm{\bx}_{\bA}^2 = \bx^{H} \bA \bx$ is a weighted norm. Finally, replacing the sum in the denominator by its maximum value~$\levb{\abs{\calS}}$ gives the desired lower bound in~\eqref{E:mseBound}.\qedhere

\end{proof}

\begin{figure}[t]
	\centering
	\includesvg{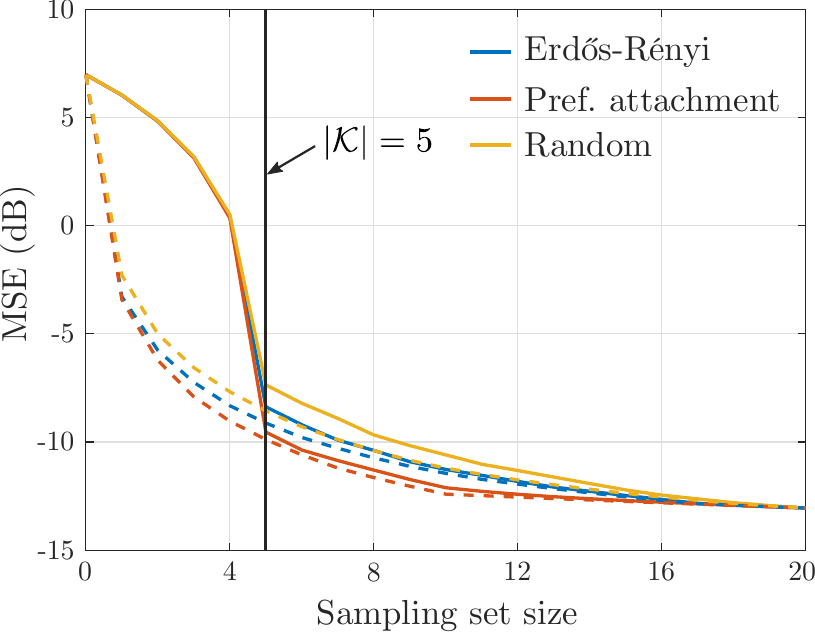}
\caption{Comparison between~\eqref{E:mseBound}~(dashed lines) and minimum MSE~(solid lines) for reconstructing graph signals~($\bH = \bI$) on random graphs~($n = 20$)}
	\label{F:MSEvsBound1}
\end{figure}

The bounds in~\eqref{E:mseBound} were derived from a Bayesian perspective, so that the expectation in the MSE is taken over realizations of the signal and noise and the bounds hold for all sampling sets~$\calS \subseteq \calV$. Also, it is worth noting that~\eqref{E:mseBound} depends only on statistics of the graph signal~($\bLambda$, $\bLambda_w$, and~$\calK$), the transform~($\bH$), the structural properties of the underlying graph~($\bV$), and the sampling set size~($\abs{\calS}$). These are all quantities known \emph{a priori}, i.e., before the sampling occurs.

As expected, \eqref{E:mseBound} decreases with the sampling set size. The rate of decay, however, depends on the \emph{weighted structural SNRs}~$\{\ell_i\}$. These quantities represent the relation between the signal of interest and the noise at each node, taking into account the structure of the graph and the subspace of interest~[$\colspan(\bV_\calK)$]. Moreover, they are related to statistical estimates such as the leverage score and the Mahalanobis distance in regression. In a sequential sampling scheme, their value can be used to inform whether a new sample is worth acquiring by bounding the MSE improvement. A bound on the decay rate can also be obtained using the fact that~$\levb{m} \leq m \ell_{\text{max}}$ for $\ell_{\text{max}} = \max_i \ell_i$. Then, if the sampling set is chosen so as to uniquely determine the graph signal, i.e., $\abs{\calS} \geq \abs{\calK}$, \eqref{E:mseBound} reduces to
\begin{equation}\label{E:kBound}
	\MSE(\calS) \geq
		\frac{\abs{\calK}}{\lambda_\text{min} (\bW \bLambda)^{-1} + \ell_{\text{max}}}
		\text{,}
\end{equation}
where we used the fact that~$\trace \left[ (\bW \bLambda)^{-1} \right] \leq \abs{\calK} \lambda_\text{min} (\bW \bLambda)^{-1}$. It is clear from~\eqref{E:kBound} that the reconstruction error increases linearly with the bandwidth of the graph signal, which is a fundamental limitation for large dimensional signals. It also shows the importance of working with low bandwidth signals and, consequently, of appropriately identifying the signal's underlying graph.

Although these observations give insights into graph signal interpolation, one of the main motivation behind Theorem~\ref{T:mseBound} is addressing the issue of sampling set selection. Towards this end, we propose the following corollary:

\begin{corollary}\label{T:sBound}
	For any graph signal and its interpolation as in Theorem~\ref{T:mseBound}, all sampling set~$\calS$ for which~$\MSE(\calS) \leq \eta$ satisfy
	\begin{equation}\label{E:sBound}
		\levb{\abs{\calS}}
		\geq
		\frac{\abs{\calK}^2 - \eta \trace \left[ (\bW \bLambda)^{-1} \right]}{\eta}
			\text{.}
	\end{equation}
	Since~$\levb{\abs{\calS}} \leq \abs{\calS} \ell_{\text{max}}$, it also holds that
	\begin{equation}\label{E:sBoundSimple}
		\abs{\calS}
		\geq
		\frac{\abs{\calK}^2 - \eta \trace \left[ (\bW \bLambda)^{-1} \right]}{\eta \ell_{\text{max}}}
			\text{.}
	\end{equation}
\end{corollary}

\begin{figure}[t]
	\centering
	\includesvg{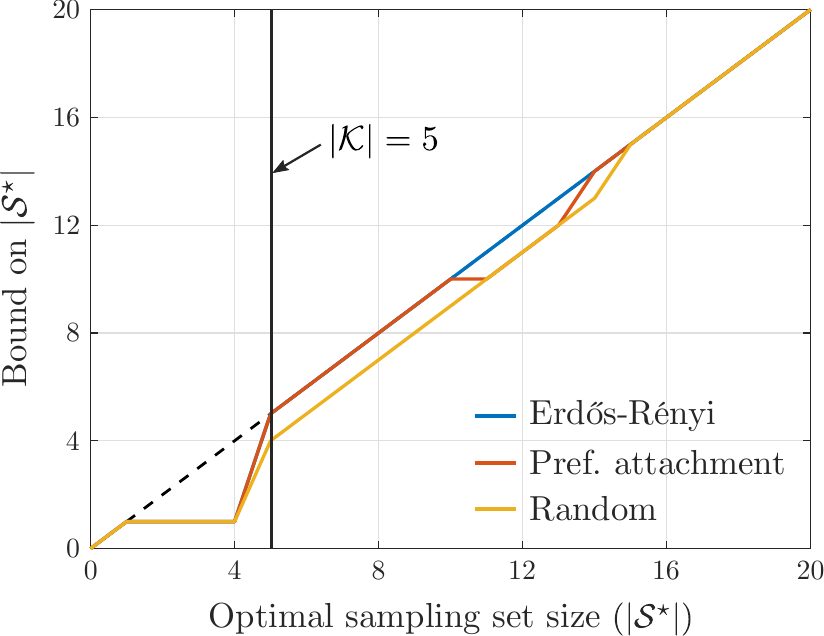}
\caption{Comparison between~\eqref{E:sBound}~(dashed lines) and optimal sampling set size~(solid lines) for reconstructing graph signals~($\bH = \bI$) on random graphs~($n = 20$)}
	\label{F:SvsBound1}
\end{figure}

Corollary~\ref{T:sBound} gives a lower bound on the number of samples needed to achieve a desired MSE. It is worth noting that this bound does not inform whether the specific MSE value~$\eta$ is achievable with less than~$n$ samples. Whenever it is not, Corollary~\ref{T:sBound} holds trivially. From~\eqref{E:sBoundSimple}, note that the minimum number of samples increases as the MSE decreases. Moreover, although~\eqref{E:sBoundSimple} suggest that the sample set size required to achieve a certain MSE grows with $\calO(\abs{\calK}^2)$, it is not necessarily the case. Indeed, recall that $\ell_{\text{max}}$ is a function of $\abs{\calK}$~(through $\norm{\bv_i}$ and~$\bV_\calK$). Still, as in the noiseless case, the signal bandwidth is a dominating factor in the determination of the minimum sampling set size.

Although~\eqref{E:sBoundSimple} characterizes the overall behavior of the sampling set size, it is not informative in practice because it largely underestimates~$\abs{\calS}$. On the other hand, \eqref{E:sBound} yields a tighter bound which can be used, together with~\eqref{E:mseBound}, to evaluate a sampling set or sampling technique for direct reconstruction of a graph signal~($\bH = \bI$). Indeed, Figures~\ref{F:MSEvsBound1} and~\ref{F:SvsBound1} compares~\eqref{E:mseBound} and~\eqref{E:sBound} to the minimum interpolation MSE and optimal set size, found by exhaustive search, for three graph models~($n = 20$): Erd\H{o}s-Rényi, preferential attachment, and a random undirected graph with weights uniformly distributed in~$[0,1]$~(see details of these models in Section~\ref{S:Simulations}). The graph signal is assumed to be homeoscedastic with~$\bLambda = \bI$ and~$\bLambda_w = \sigma_w^2 \bI$, $\sigma_w^2 = 10^{-2}$. Note that the bounds are conservative for $\abs{\calS} < \abs{\calK}$, but become tighter as~$\abs{\calS}$ increases. This is because the inequality used to derive~\eqref{E:mseBound} becomes tighter as the eigenvalues of~$\bK^\star$ become more similar.

When~$\bH$ is arbitrary and variance of signal and noise can vary across nodes, the eigenvalues of~$\bK^\star$ can become different from each another and deteriorate the bound in~\eqref{E:mseBound}. This is illustrated in Figures~\ref{F:MSEvsBound2} and~\ref{F:SvsBound2}, where~$\bH$ was taken as a~$30 \times 20$ matrix whose entries are zero-mean unit variance Gaussian random variables, $\bLambda = \bI$, and the noise variance was uniformly distributed in~$[10^{-3}, 10^{-1}]$.

\begin{remark}

Bounds on the interpolation MSE of graph sampling techniques have also been derived in~\cite{Chen15d, Chen16s}. These works consider randomized sampling set selection schemes, including uniform and leverage score sampling, and derive performance bounds on the optimal sampling distributions and interpolation error. The bounds in Theorem~\ref{T:mseBound} and Corollary~\ref{T:sBound} differ from those in~\cite{Chen15d, Chen16s} in that the latter take the spectrum of the graph signal to be deterministic and the sampling to be random. Thus, these bounds hold in expectation over different sampling realizations for a specific randomized strategy. The bounds in~\eqref{E:mseBound} hold in expectation over realizations of the signal and noise and give a worst-case performance bound for any---possibly randomized---sampling strategy.

\end{remark}

\begin{figure}[t]
	\centering
	\includesvg{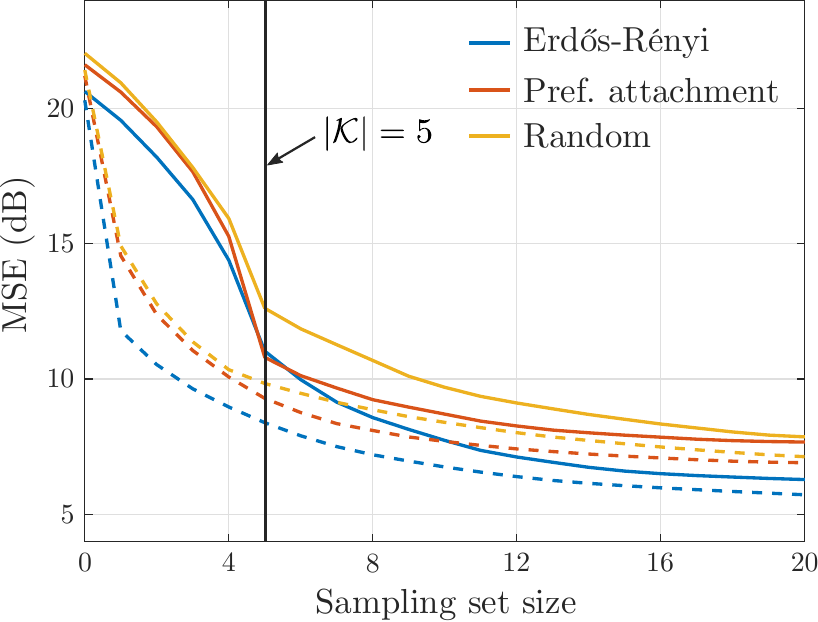}
\caption{Comparison between~\eqref{E:mseBound}~(dashed lines) and minimum MSE~(solid lines) for a random~$\bH$~($n = 20$)}
	\label{F:MSEvsBound2}
\end{figure}

\begin{figure}[t]
	\centering
	\includesvg{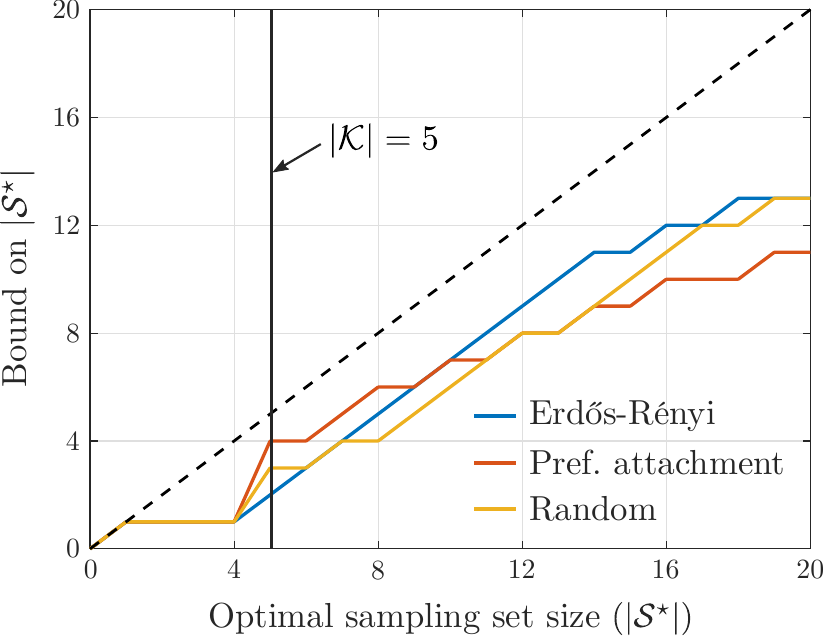}
\caption{Comparison between~\eqref{E:sBound}~(dashed lines) and optimal sampling set size~(solid lines) for a random~$\bH$~($n = 20$)}
	\label{F:SvsBound2}
\end{figure}

\section{Near-Optimal Sampling Set Selection}
	\label{S:greedySampling}

Although the bounds from Section~\ref{S:Bounds} can be used to evaluate specific sampling set instances, they do not provide performance guarantees for any sampling strategy. To do that, this section studies a specific sampling scheme, namely \emph{greedy sampling set selection}, and derives near-optimality results that hold for all problem instances.

Greedy sampling set selection is pervasive in GSP and has proven successful in many applications~\cite{Chen15d, Anis16e, Shomorony14s, Thanou14l, Tsitsvero16s, Chepuri16s}. This is illustrated in Figure~\ref{F:GreedyvsBound} which uses the bounds derived in~\eqref{E:sBound} to assess the quality of sampling sets obtained by greedily minimizing the MSE~(see Algorithm~\ref{L:greedyMSE}) on larger instances~($n = 1000$) of the three random graph models found in Figures~\ref{F:MSEvsBound1} to~\ref{F:SvsBound2}. Note that the final greedy sampling set size remains within 10\% of the lower bound in these realizations.

\begin{figure}[tb]
	\centering
	\includesvg{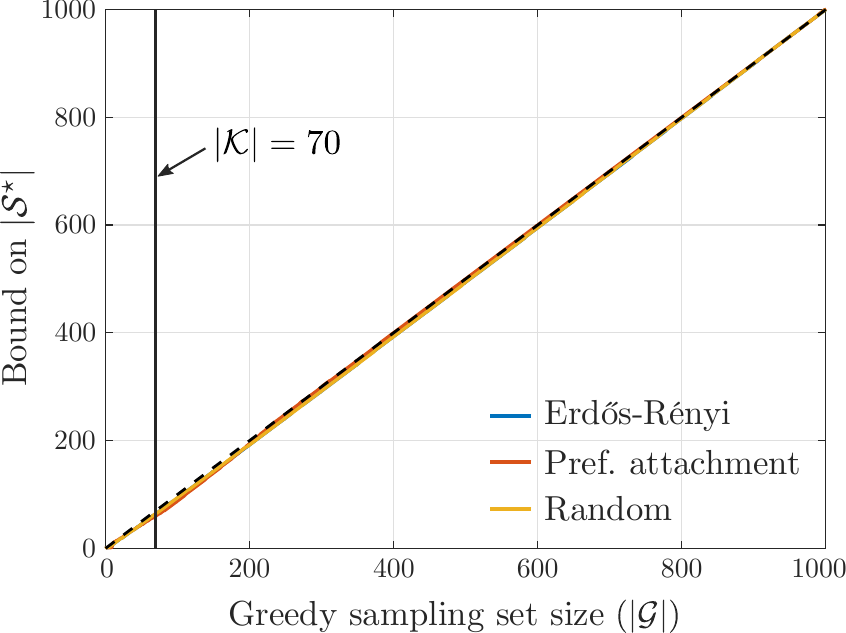}
\caption{Evaluating sampling sets obtained by greedy sampling~(solid lines) and~\eqref{E:sBound}~(dashed lines) for different random graphs~($n = 1000$)}
	\label{F:GreedyvsBound}
\end{figure}

Despite strong empirical evidences, typical performance guarantees for greedy search do not hold for greedy sampling set selection. Indeed, the well-established result from~\cite{Nemhauser78a} states that greedy minimization~(Algorithm~\ref{L:greedy}) yields guaranteed near-optimal results for monotonically decreasing and supermodular set functions. The MSE, however, is not supermodular in general. This can be seen from~\cite[Thm.~2.4]{Sagnol13a} and the fact that $f(t) = t^{-2}$ is not operator antitone~\cite{Bhatia97m}. Thus, although greedily minimizing the MSE appears to work in practice, there has yet to be a theoretical justification for it. The following sections bridge this gap by expanding the notion of approximate supermodularity introduced in~\cite{Chamon16n} and updating the performance bound from~\cite{Nemhauser78a} for this class of functions. This novel framework then allows near-optimality bounds to be derived for the MSE.

\begin{algorithm}[t]
\centering
\caption{Greedy minimization}
	\label{L:greedy}
	\setlength{\baselineskip}{1.25\baselineskip}
	\begin{algorithmic}
		\State $\calG_0 = \{\}$

		\For{$j = 1,\dots,\ell$}

			\State $\displaystyle
						u = \argmin_{s \in \calV \setminus \calG_{j-1}} f\left( \calG_{j-1} \cup \{s\} \right)$

			\State $\displaystyle
						\calG_j = \calG_{j-1} \cup \{u\}$
		\EndFor
	\end{algorithmic}
\end{algorithm}

\subsection{Approximate supermodularity and greedy minimization}
	\label{S:Supermodularity}

\emph{Supermodularity}~(and its dual \emph{submodularity}) encodes the ``diminishing returns'' property of certain functions that leads to bounds on the suboptimality of their greedy minimization~\cite{Nemhauser78a}. Well-known supermodular functions include the rank, $\log\det$, or Von-Neumann entropy of a matrix~\cite{Bach14l}. Still, supermodularity is a stringent condition. In particular, it does not hold for the MSE in~\eqref{P:SSS}. To provide suboptimality bounds for its greedy minimization, we therefore define the concept of approximate supermodularity.

A set function $f: 2^\calV \to \setR$ is~\emph{$\alpha$-supermodular} if for all sets $\calA \subseteq \calB \subseteq \calV$ and all $u \notin \calB$ it holds that
\begin{equation}\label{E:supermodularity}
	f\left( \calA \cup \{u\} \right) - f\left( \calA \right)
	\leq
	\alpha \left[ f\left( \calB \cup \{u\} \right) - f\left( \calB \right) \right]
		\text{,}
\end{equation}
for $\alpha \geq 0$. We say~$f$ is~\emph{$\alpha$-submodular} if $-f$ is $\alpha$-supermodular. For~$\alpha \geq 1$, \eqref{E:supermodularity} is equivalent to the traditional definition of supermodularity, in which case we refer to the function simply as \emph{supermodular}/\emph{submodular}~\cite{Bach14l}. For~$\alpha \in [0,1)$, however, $f$ is said to be \emph{approximately supermodular}/\emph{submodular}. Notice that~\eqref{E:supermodularity} always holds for $\alpha = 0$ if $f$ is monotone decreasing. Indeed, $f\left( \calA \cup \{u\} \right) - f\left( \calA \right) \leq 0$ in this case. Thus, $\alpha$-supermodularity is only of interest when $\alpha$ takes the largest value for which~\eqref{E:supermodularity} holds, i.e.,
\begin{equation}\label{E:alpha}
	\alpha = \min_{\substack{\calA \subseteq \calB \subseteq \calV \\ u \notin \calB}}
		\frac{f(\calA \cup \{u\}) - f(\calA)}{f(\calB \cup \{u\}) - f(\calB)}
		\text{.}
\end{equation}

Before proceeding, it is worth noting that~$\alpha$ is related to the \emph{submodularity ratio} from~\cite{Das11s}. It is, however, more amenable to give explicit bounds on its value~(see Section~\ref{S:alphaBounds}). The submodularity ratio bounds derived in~\cite{Das11s} depend on the minimum sparse eigenvalue of a matrix, which cannot be evaluated efficiently. Due to the relation between~$\alpha$ and the submodularity ratio, it is not surprising that similar near-optimal bounds hold for $\alpha$-supermodular functions.

\begin{theorem}
	\label{T:greedy}
	Let $f^\star = f(\calS^\star)$ be the optimal value of the problem
	\begin{equation}\label{E:greedyOptim}
		\underset{\calS \subseteq \calV \text{, } \abs{\calS} \leq k}{\textup{minimize}}
		\quad
		f\left( \calS \right)
	\end{equation}
	and~$\calG_\ell$ be the set obtained by applying Algorithm~\ref{L:greedy}. If $f$ is (i)~monotone decreasing and (ii)~$\alpha$-supermodular, then
	\begin{equation}\label{E:relOptimality}
		\frac{f(\calG_\ell) - f^\star}{f(\{\}) - f^\star}
			\leq \left( 1 - \frac{\alpha}{k} \right)^{\ell}
			\leq e^{-\alpha \ell/k}
			\text{,}
	\end{equation}
	where~$f(\{\})$ is the value of function for the empty set. If $f$ is normalized, i.e., $f(\{\}) = 0$, \eqref{E:relOptimality} reduces to
	\begin{equation*}
		f(\calG_\ell) \leq \left( 1 - e^{-\alpha \ell/k} \right) f^\star
			\text{.}
	\end{equation*}

\end{theorem}

\begin{proof}
	Using the fact that $f$ is monotone decreasing, it holds for every set $\calG_j$ that
	\begin{equation*}
		f(\calS^\star) \geq f(\calS^\star \cup \calG_j)
			\text{.}
	\end{equation*}
	Using a telescopic sum then gives
	\begin{equation}\label{E:greedyTelescopic}
		f(\calS^\star) \geq f(\calG_j) +
			\sum_{i = 1}^{k} f(\calT_{i-1} \cup \{s_i^\star\}) - f(\calT_{i-1})
			\text{,}
	\end{equation}
	where $\calT_{i} = \calG_j \cup \{ s_{1}^\star, \dots, s_{i}^\star \}$ and $s_i^\star$ is the $i$-th element of~$\calS^\star$. Since $f$ is $\alpha$-supermodular and $\calG_j \subseteq \calT_i$ for all $i$, the incremental gains in the summation in~\eqref{E:greedyTelescopic} can be bounded using~\eqref{E:supermodularity} to get
	\begin{equation*}
		f(\calS^\star) \geq f(\calG_j) +
			\alpha^{-1} \sum_{i = 1}^{k} \left[ f(\calG_j \cup \{s_i^\star\}) - f(\calG_j) \right]
			\text{.}
	\end{equation*}
	Finally, given that $\calG_{j+1} = \calG_j \cup \{u\}$ is chosen to minimize~$f(\calG_{j+1})$~(see Algorithm~\ref{L:greedy}),
	\begin{equation}\label{E:greedyRecursion1}
		f(\calS^\star) \geq f(\calG_j) +
			\alpha^{-1} k \left[ f(\calG_{j+1}) - f(\calG_j) \right]
			\text{.}
	\end{equation}

	To obtain a recursion, let $\delta_j = f(\calG_j) - f(\calS^\star)$ so that~\eqref{E:greedyRecursion1} becomes
	\begin{align*}
		\delta_j \leq \alpha^{-1} k \left[ \delta_{j} - \delta_{j+1} \right]
		\Rightarrow
		\delta_{j+1} \leq \left( 1 - \frac{1}{\alpha^{-1} k} \right) \delta_{j}
			\text{.}
	\end{align*}
	Noting that $\delta_0 = f(\{\}) - f(\calS^\star)$, we can solve this recursion to get
	\begin{align*}
		\frac{f(\calG_\ell) - f(\calS^\star)}{f(\{\}) - f(\calS^\star)} \leq
		\left( 1 - \frac{\alpha}{k} \right)^{\ell}
			\text{.}
	\end{align*}
	Using the fact that $1 - x \leq e^{-x}$ yields~\eqref{E:relOptimality}.
\end{proof}

Theorem~\ref{T:greedy} bounds the relative suboptimality of the greedy solution to problem~\eqref{E:greedyOptim} when~$f$ is decreasing and $\alpha$-supermodular. Under these conditions, it guarantees a minimum improvement of the greedy solution over the empty set. What is more, it quantifies the effect of relaxing the supermodularity hypothesis in~\eqref{E:supermodularity}. Indeed, when $f$ is supermodular~($\alpha = 1$) and the greedy search in Algorithm~\ref{L:greedy} is repeated~$k$ times~($\ell = k$), we recover the $e^{-1} \approx 0.37$ guarantee from~\cite{Nemhauser78a}. On the other hand, if~$f$ is not supermodular~($\alpha < 1$), \eqref{E:relOptimality} shows that the same $37\%$ guarantee can be obtained by greedily selecting a set of size~$\alpha^{-1} k$. Thus, $\alpha$ not only quantifies how much $f$ violates supermodularity, but also gives a factor by which a solution set must increase to maintain supermodular near-optimality. In other words, it measures the constraint violation needed to recover the~$37\%$ guarantee. It is worth noting that, as with the original bound in~\cite{Nemhauser78a}, \eqref{E:relOptimality} is not tight and that better results are common in practice~(see Section~\ref{S:Simulations}).

In the sequel, we show that~$\MSE(\calS)$ is a monotone decreasing and $\alpha$-supermodular function of the sampling set~$\calS$. We also provide an explicit lower bound on~$\alpha$ as a function of the SNR. This result simultaneously provide near-optimal performance guarantees based on Theorem~\ref{T:greedy} and sheds light on why greedy algorithms have been so successful in GSP applications.

\subsection{Near-optimality of greedy sampling set selection}
	\label{S:alphaBounds}

\begin{figure}[tb]
\centering
\includesvg{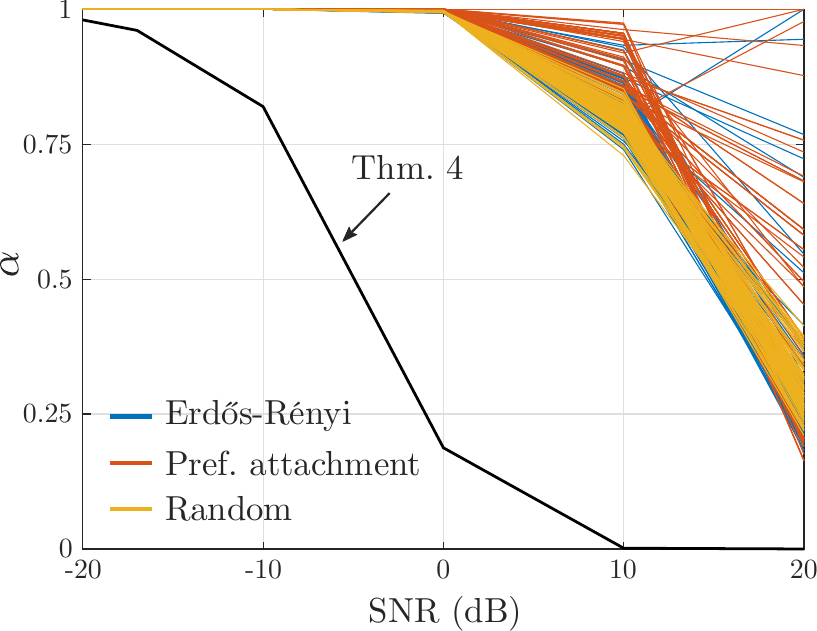}
\caption{Comparison between the bound in~\eqref{E:mainAlpha2} and $\alpha$}
	\label{F:alpha}
\end{figure}

The main result of this section is:

\begin{theorem}\label{T:main}

Let~$\calS^\star$ be the solution of~\eqref{P:SSS} and~$\calG_\ell$ be the result of applying Algorithm~\ref{L:greedy} for~$f(\calS) = \MSE(\calS)$. Then, 
\begin{equation*}
	\frac{\MSE(\calG_\ell) - \MSE(\calS^\star)}{\MSE(\{\}) - \MSE(\calS^\star)}
		\leq e^{-\alpha \ell/k}
		\text{,}
\end{equation*}
where
\begin{equation}\label{E:mainAlpha1}
\alpha \geq \frac{
		\lambda_\textup{max}(\bLambda_w)^{-1} + \mu_\textup{max}^{-1}
	}{
		\lambda_\textup{max}(\bLambda_w)^{-1} + \mu_\textup{min}^{-1}
	}
	\frac{
		\mu_\textup{min}^{2}
	}{
		\kappa_2(\bW) \, \mu_\textup{max}^{2}
	}
\end{equation}
for~$\mu_\textup{min} \leq \lambda_\textup{min} \left[ \bLambda^{-1} \right]$, $\mu_\textup{max} \geq \lambda_\textup{max} \left[ \bLambda^{-1} + \bV_\calK^H \bLambda_{w}^{-1} \bV_\calK \right]$, and~$\kappa_2(\bW)$ is the 2-norm condition number of~$\bW$. Assuming $\bLambda = \sigma_x^2 \bI$ and~$\bLambda_w = \sigma_w^2 \bI$, \eqref{E:mainAlpha1} reduces to
\begin{equation}\label{E:mainAlpha2}
	\alpha \geq	\frac{1 + 2 \gamma}{\kappa_2(\bW) \, (1 + \gamma)^{4}}
		\text{,} \quad  \text{for}\ \ \gamma = \frac{\sigma_x^2}{\sigma_w^2}
		\text{.}
\end{equation}
\end{theorem}

Theorem~\ref{T:main} establishes that a near-optimal solution to the sampling set selection problem in~\eqref{P:SSS} can be obtained efficiently using greedy search. Though strong empirical evidence exists that greedily minimizing the MSE yields good results in contexts such as regression, dictionary learning, and graph signal processing~\cite{Das11s, Chen15d, Anis16e, Shomorony14s, Thanou14l, Tsitsvero16s}, this result is counter-intuitive given that the MSE is not supermodular in general. For instance, restrictive and often unrealistic conditions on data distribution are required to obtain supermodularity in the context of regression~\cite{Das11s}.

Theorem~\ref{T:main} therefore reconciles the empirical success of greedy sampling set selection and the non-supermodularity of the MSE by bounding the suboptimality of greedy sampling. In particular, \eqref{E:mainAlpha2} gives a simple bound on~$\alpha$ in terms of the SNR and the condition number of~$\bW$ that gives clear insights into its behavior. Indeed, as $\gamma \to \infty$ and we approach the noiseless case, $\alpha \to 0$. This is expected as in the noiseless case almost every set of size~$\abs{\calK}$ achieves perfect reconstruction, so that the choice of sampling nodes is irrelevant. On the other hand, $\alpha \to 1$ as~$\gamma \to 0$, i.e., the MSE becomes closer to supermodular as the SNR decreases. Given that reconstruction errors are small for high SNR, Theorem~\ref{T:main} guarantees that greedy sampling performs well when it is most needed. Similar trends can be observed in the more general setting of~\eqref{E:mainAlpha1}. These observations are illustrated in Figure~\ref{F:alpha} that compares the bound in~\eqref{E:mainAlpha2} to the true value of~$\alpha$ for the MSE~(found by exhaustive search) in $100$~realizations of random graphs~(see Section~\ref{S:Simulations} for details).

Theorem~\ref{T:main} stems directly from Theorem~\ref{T:greedy} and the following characterizations of the MSE function:

\begin{lemma}\label{T:MSEFull}

The scalar set functions $\MSE(\calS) = \trace[\bK^\star(\calS)]$ is (i)~monotone decreasing and (ii)~$\alpha$-supermodular with
\begin{equation}\label{E:alphaMSEFull}
\alpha \geq \frac{
		\lambda_\textup{max}(\bLambda_w)^{-1} + \mu_\textup{max}^{-1}
	}{
		\lambda_\textup{max}(\bLambda_w)^{-1} + \mu_\textup{min}^{-1}
	}
	\frac{
		\mu_\textup{min}^{2}
	}{
		\kappa_2(\bW) \, \mu_\textup{max}^{2}
	}
		\text{,}
\end{equation}
where~$\mu_\textup{min} \leq \lambda_\textup{min} \left[ \bLambda^{-1} \right]$, $\mu_\textup{max} \geq \lambda_\textup{max} \left[ \bLambda^{-1} + \bV_\calK^H \bLambda_{w}^{-1} \bV_\calK \right]$, and~$\kappa_2(\bW)$ is the 2-norm condition number of~$\bW$~\cite{Horn13}.

\end{lemma}

\begin{lemma}\label{T:MSE}

Assuming $\bLambda = \sigma_x^2 \bI$ and~$\bLambda_w = \sigma_w^2 \bI$, the set functions $\MSE(\calS) = \trace[\bK^\star(\calS)]$ is (i)~monotone decreasing and (ii)~$\alpha$-supermodular with
\begin{equation}\label{E:alphaMSE}
	\alpha \geq	\frac{1 + 2 \gamma}{\kappa_2(\bW) \, (1 + \gamma)^{4}}
		\text{,} \quad  \text{for}\ \ \gamma = \frac{\sigma_x^2}{\sigma_w^2}
		\text{,}
\end{equation}
where~$\kappa_2(\bW)$ is the 2-norm condition number of~$\bW$~\cite{Horn13}.

\end{lemma}

The proof of Lemma~\ref{T:MSEFull} is deferred to Appendix~\ref{A:ThmMseFull}. Here, we proceed with the proof of Lemma~\ref{T:MSE} after stating a pertinent remark.

\begin{remark}

Since the MSE is \emph{not} supermodular, it is common to see surrogate supermodular figures of merit used instead, specially in statistics and experiment design~\cite{Bach14l, Das11s, Sagnol13a, Krause08n}. In particular, the log-determinant~$\log\det[\bK^\star(\calS)]$ is a common alternative to the objective~$\MSE(\calS) = \trace[\bK^\star(\calS)]$ used in~\eqref{P:SSS}. This is justified because the~$\log\det[\bK^\star(\calS)]$ is proportional to the volume of the confidence ellipsoids of the estimate when the data is Gaussian~\cite{Joshi09s, Sagnol13a}. This choice of objective is also common in the sensor placement literature due to its relation to information theoretic measures, such as entropy and mutual information~\cite{Krause08n}. By replacing the trace operator in~\eqref{P:SSS} by the~$\log\det$, the problem becomes a supermodular function minimization that can be efficiently approximated using greedy search, as shown in~\cite{Chamon16n, Chepuri16s}. We remark that minimizing the~$\log\det$ of the error covariance matrix and the MSE are not equivalent problems. 

\end{remark}

\subsection{Proof of Lemma~\ref{T:MSE}}

	Start by noticing that part~(i) stems directly from Lemma~\ref{T:monotonicity}. Indeed, the monotonicity of the trace implies that~$\bX \succeq \bY \Rightarrow \trace (\bX) \geq \trace (\bY)$, for any PSD matrices $\bX$ and~$\bY$.
	
	Then, to obtain part~(ii), use the homeoscedasticity assumption to rewrite~\eqref{E:Kstar} as
\begin{equation}\label{E:KS2}
	\bK^\star(\calS) = \sigma_x^2 \bH \bV_\calK \bZ(\calS)^{-1} \bV_\calK^{H} \bH^{H}
		\text{,}
\end{equation}
where~$\bZ(\calS) = \bI + \gamma \sum_{i \in \calS} \bv_i^{} \bv_i^{H}$ and~$\gamma = \sigma_x^2/\sigma_w^2$ is the SNR. Then, proceed to obtain a closed form expression for the increments in~\eqref{E:alpha} by using~\eqref{E:KS2} to get
	\begin{multline*}
		f(\calA \cup \{u\}) - f(\calA) ={}
		\\
		\sigma_x^2 \trace \left[
			\bW \left( \bZ(\calA) + \gamma \bv_u^{} \bv_u^{H} \right)^{-1}
			- \bW \bZ(\calA)^{-1}
		\right]
			\text{.}
	\end{multline*}
	From the matrix inversion lemma~\cite{Horn13}, this expression reduces to
	\begin{equation*}
		f(\calA \cup \{u\}) - f(\calA) =
		- \sigma_x^2 \trace \left[
			\bW \frac{
				\bZ(\calA)^{-1} \bv_u \bv_u^{H} \bZ(\calA)^{-1}
			}{
				\gamma^{-1} + \bv_u^{H} \bZ(\calA)^{-1} \bv_u
			}
		\right]
			\text{,}
	\end{equation*}
	which using the commutation property of the trace yields
	\begin{equation}\label{E:mseIncrement}
		f(\calA \cup \{u\}) - f(\calA) =
		-\sigma_x^2 \frac{
			\bv_u^{H} \bZ(\calA)^{-1} \bW \bZ(\calA)^{-1} \bv_u
		}{
			\gamma^{-1} + \bv_u^{H} \bZ(\calA)^{-1} \bv_u
		}
			\text{.}
	\end{equation}
	From~\eqref{E:mseIncrement}, the expression for~$\alpha$ in~\eqref{E:alpha} becomes
	\begin{equation}\label{E:alpha2}
	\alpha = \min_{\substack{\calA \subseteq \calB \subseteq \calV \\ u \notin \calB}}
			\frac{
				\gamma^{-1} + \bv_u^{H} \bZ(\calB)^{-1} \bv_u
			}{
				\gamma^{-1} + \bv_u^{H} \bZ(\calA)^{-1} \bv_u
			}
			\frac{
				\bv_u^{H} \bZ(\calA)^{-1} \bW \bZ(\calA)^{-1} \bv_u
			}{
				\bv_u^{H} \bZ(\calB)^{-1} \bW \bZ(\calB)^{-1} \bv_u
			}
		\text{.}
	\end{equation}

	To bound~\eqref{E:alpha2}, first notice that for any set $\calX \subseteq \calV$
	\begin{equation}
		1 \leq \lambda_\text{min}\left[ \bZ(\calX) \right] \leq
		\lambda_\text{max}\left[ \bZ(\calX) \right] \leq 1+\gamma 
			\text{,}
	\end{equation}
	where $\lambda_\text{min}$ and $\lambda_\text{max}$ denote the minimum and maximum eigenvalues of a matrix. These bounds are achieved for the empty set and~$\calV$, respectively. Then, using the Rayleigh quotient inequalities~\cite{Horn13}
	\begin{equation*}
		\norm{\bb}_2^2 \lambda_\text{min}(\bA) \leq \bb^H \bA \bb \leq \norm{\bb}_2^2 \lambda_\text{max}(\bA)
			\text{,}
	\end{equation*}
	we get that~\eqref{E:alpha2} is bounded by
	\begin{equation*}
	\alpha \geq \frac{\gamma^{-1} + \norm{\bv_u}_2^2 (1+\gamma)^{-1}}{\gamma^{-1} + \norm{\bv_u}_2^2}
		\cdot
		\frac{
			\lambda_\textup{min} \left[ \bZ(\calA)^{-1} \bW \bZ(\calA)^{-1} \right]
		}{
			\lambda_\textup{max} \left[ \bZ(\calB)^{-1} \bW \bZ(\calB)^{-1} \right]
		}
		\text{.}
	\end{equation*}
	To simplify this expression, let~$\sigma_i(\bA)$ denote the $i$-th singular value of~$\bA$ and recall that for~$\bA \succeq 0$, $\sigma_i^2(\bA) = \lambda_i(\bA\bA^H)$. Thus, we can write~$\lambda_i \left[ \bZ(\calA)^{-1} \bW \bZ(\calA)^{-1} \right] = \sigma_i^2 \left[ \bZ(\calA)^{-1} \bW^{1/2} \right]$, which is well-defined since~$\bW \succeq 0$. Using the fact that~$\sigma_\text{max}(\bA \bB) \leq \sigma_\text{max}(\bA) \sigma_\text{max}(\bB)$ and~$\sigma_\text{min}(\bA \bB) \geq \sigma_\text{min}(\bA) \sigma_\text{min}(\bB)$~\cite[Thm.~9.H.1, p.~338]{Marshall09i}, we obtain
	\begin{equation}\label{E:alpha3}
	\alpha \geq \frac{\gamma^{-1} + 1 + \norm{\bv_u}_2^2}{\gamma^{-1} + \norm{\bv_u}_2^2}
		\cdot
		\frac{(1+\gamma)^{-3}}{\kappa_2(\bW)}
	\triangleq \alpha^\prime
		\text{.}
	\end{equation}
	where $\kappa_2(\bW) = \lambda_\text{max}(\bW) / \lambda_\text{min}(\bW)$ is the 2-norm condition number of $\bW$~\cite{Horn13}.
	
	Finally, to obtain the expression in Lemma~\ref{T:MSE}, notice that~\eqref{E:alpha3} is decreasing with respect to~$\norm{\bv_u}_2^2$. Indeed, since~$\kappa_2 \geq 1$ and~$\gamma \geq 0$,
	\begin{equation*}
		\frac{\partial \alpha^\prime}{\partial \norm{\bv_u}_2^2} =
		\frac{
			- (1+\gamma)^{-3}
		}{
			\kappa_2(\bW) \left( \gamma^{-1} + \norm{\bv_u}_2^2 \right)^2
		} \leq 0
		\text{.}
	\end{equation*}
	Given that~$\bv_u$ is a row of~$\bV_\calK$, i.e., it is composed of a subset of elements from a unit vector, $\norm{\bv_u}_2^2 \leq 1$ and we obtain the result in~\eqref{E:alphaMSE}.\qedhere

\section{Numerical Examples and Applications}
	\label{S:NumericalEg}

\begin{figure*}[t]
\centering
\includesvg{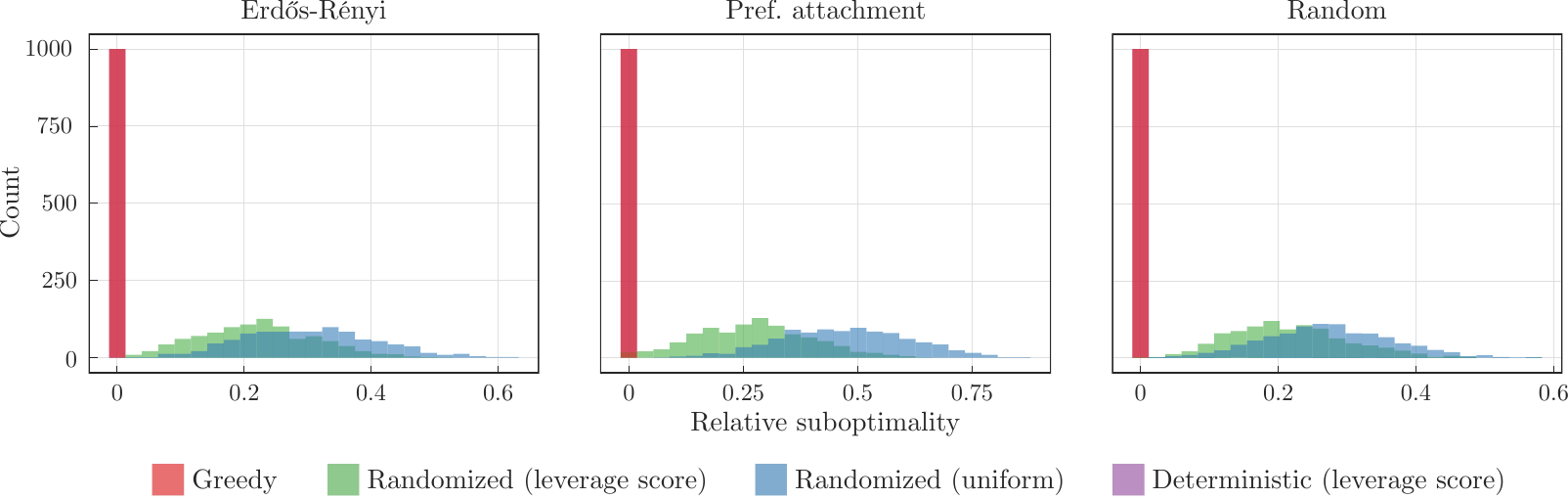}
\caption{Relative suboptimality of sampling schemes for low SNR~($\text{SNR} = -20$~dB)}
	\label{F:LowSNR}
\end{figure*}

\begin{figure*}[t]
\centering
\includesvg{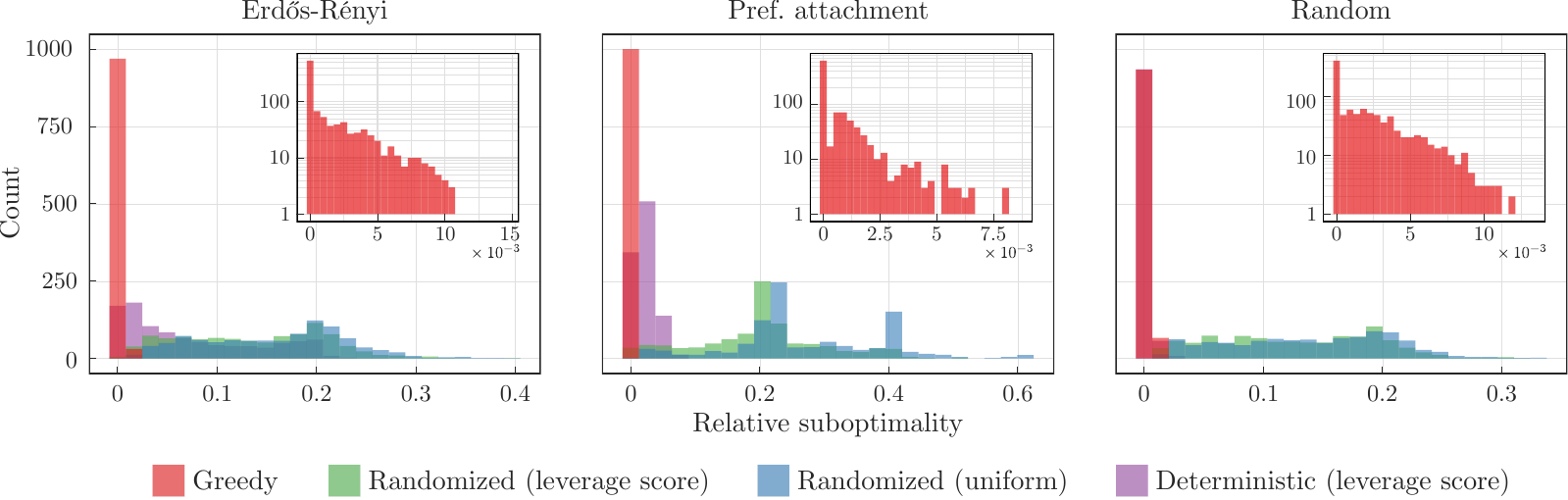}
\caption{Relative suboptimality of sampling schemes for high SNR~($\text{SNR} = 20$~dB)}
	\label{F:HighSNR}
\end{figure*}

\begin{algorithm}[t]
\centering
\caption{Greedy sampling set selection}
	\label{L:greedyMSE}
	\setlength{\baselineskip}{1.25\baselineskip}
	\begin{algorithmic}[0]
		\State $\calG_0 = \{\}$ and $\bK^\star_0 = \bLambda$

		\For{$j = 1,\dots,\ell$}

			\State $\displaystyle
						u = \argmax_{s \in \calV \setminus \calG_{j-1}}
						\frac{
							\bv_u^H \bK^\star_{j-1} \bW \bK^\star_{j-1} \bv_u
						}{
							\lambda_{w,u}^{-1} + \bv_u^H \bK^\star_{j-1} \bv_u
						}$
				\Comment{$\calO(n \abs{\calK}^2)$}

			\State $\displaystyle
						\bK^\star_{j} = \bK^\star_{j-1} - 
						\bW \frac{
							\bK^\star_{j-1} \bv_u \bv_u^H \bK^\star_{j-1}
						}{
							\lambda_{w,u}^{-1} + \bv_u^H \bK^\star_{j-1} \bv_u
						}$
				\Comment{$\calO(\abs{\calK}^2)$}

			\State $\displaystyle
						\calG_j = \calG_{j-1} \cup \{u\}$
		\EndFor
	\end{algorithmic}
\end{algorithm}

Before proceeding with the simulations, the complexity issue of greedy sampling set selection must be addressed. The greedy search in Algorithm~\ref{L:greedy} requires~$n \ell c_f$ operations, where~$c_f$ is the cost of evaluating the objective~$f$. As it is, problem~\eqref{P:SSS} has~$c_f = \calO(\abs{\calK}^3)$. It can, however, be reduced using the matrix inversion lemma~\cite{Horn13}.

Indeed, start by noticing that the first step of the greedy approximation of problem~\eqref{P:SSS} involves finding~(see Algorithm~\ref{L:greedy})
\begin{equation*}
	u = \argmin_{s \in \calV} \trace \left[
		\bK^\star \left( \calG_{j-1} \cup \{s\} \right)
	\vphantom{\sum}\right]
		\text{,}
\end{equation*}
which, using the definition of~$\bK^\star$ in~\eqref{E:Kstar} and the circular commutation property of the trace, requires the evaluation of
\begin{multline*}
	\trace \left[
		\bK^\star \left( \calG_{j-1} \cup \{s\} \right)
	\vphantom{\sum}\right] ={}
	\\
	\trace \left[ \bW \left( \bLambda^{-1} +
		\sum_{i \in \calG_{j-1}} \lambda_{w,i}^{-1} \bv_i \bv_i^{H} +
		\lambda_{w,s}^{-1} \bv_s \bv_s^{H}  \right)^{-1} \right]
		\text{,}
\end{multline*}
where once again~$\bW = \bV_\calK^{H} \bH^{H} \bH \bV_\calK$. Letting~$\bK^\star_{j} = \bK^\star(\calG_{j})$ and using the matrix inversion lemma, we can reduce the update of~$\bK^\star$ to
\begin{equation}\label{E:MILUpdate}
	\bK^\star(\calG_{j} \cup \{s\}) = 
	\bK^\star_{j-1} - 
		\bW \frac{
			\bK^\star_{j-1} \bv_u \bv_u^{H} \bK^\star_{j-1}
		}{
			\lambda_{w,u}^{-1} + \bv_u^{H} \bK^\star_{j-1} \bv_u
		}
		\text{.}
\end{equation}
From linearity, it is then straightforward to see that finding the minimum of the trace of~\eqref{E:MILUpdate} is equivalent to finding the maximum of
\begin{equation}\label{E:greedyMIL}
	\trace \left[
		\bW \frac{
			\bK^\star_{j-1} \bv_u \bv_u^{H} \bK^\star_{j-1}
		}{
			\lambda_{w,u}^{-1} + \bv_u^{H} \bK^\star_{j-1} \bv_u
		} \right] =
	\frac{
		\bv_u^{H} \bK^\star_{j-1} \bW \bK^\star_{j-1} \bv_u
	}{
		\lambda_{w,u}^{-1} + \bv_u^{H} \bK^\star_{j-1} \bv_u
	}
		\text{.}
\end{equation}
The greedy sampling set selection procedure obtained by leveraging~\eqref{E:MILUpdate} and~\eqref{E:greedyMIL} is presented in Algorithm~\ref{L:greedyMSE}. This algorithm now requires only~$\calO(n \ell \abs{\calK}^2)$ operations.

\subsection{Simulations}
	\label{S:Simulations}

\begin{figure}[t]
\centering
\includesvg[width=\columnwidth]{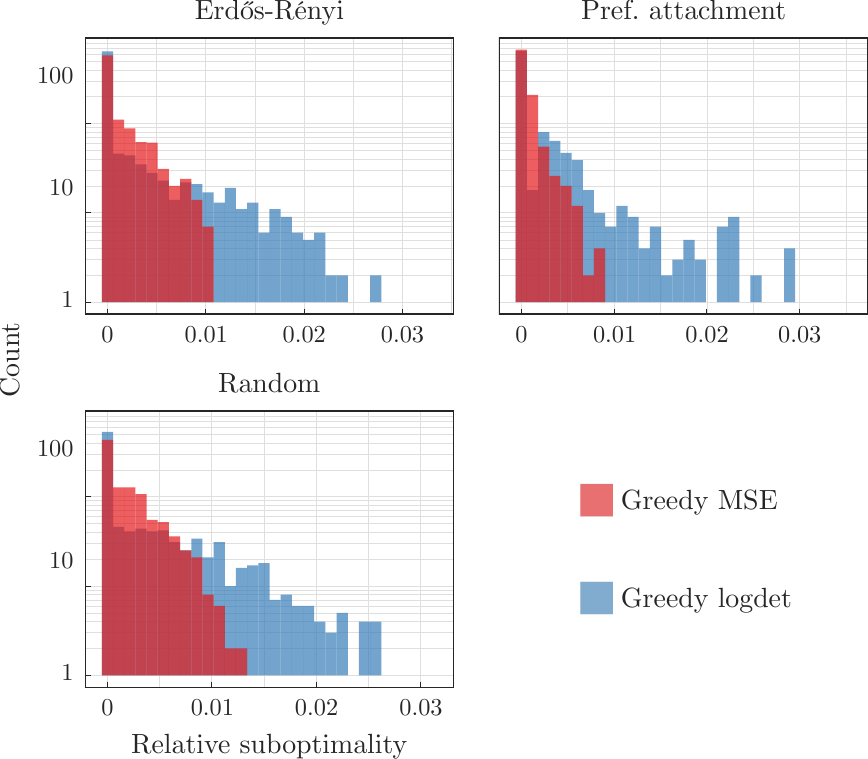}
\caption{Relative suboptimality of MSE and $\log\det$~($\text{SNR} = 20$~dB)}
	\label{F:logdet}
\end{figure}

\begin{figure*}[t]
\centering
\includesvg{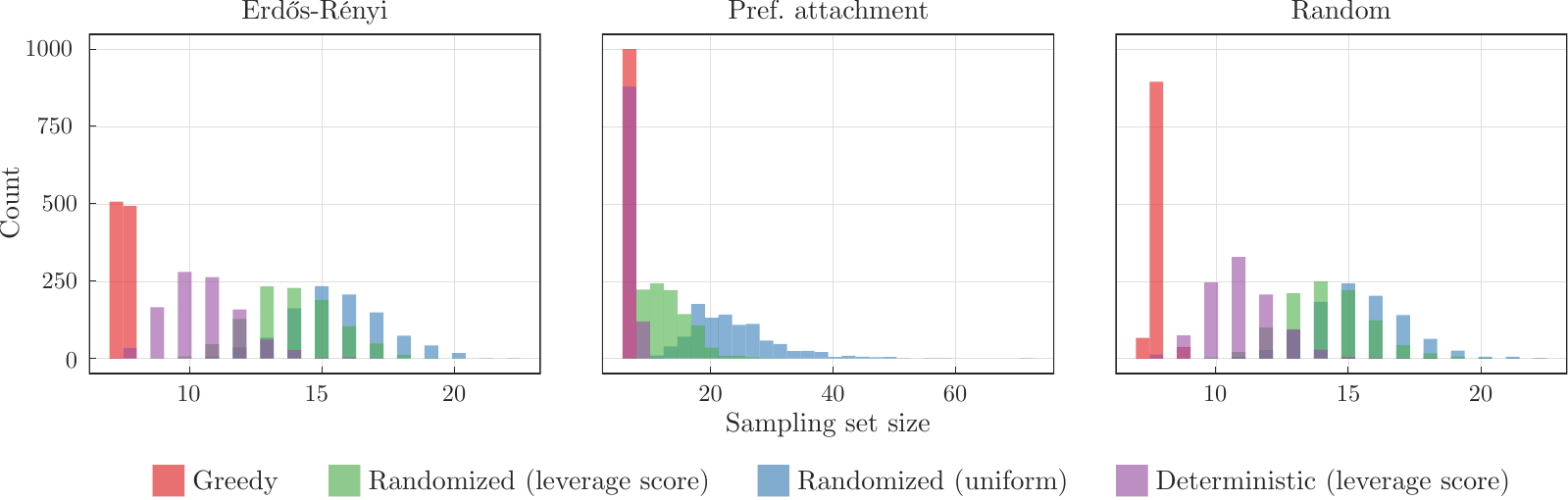}
\caption{Sampling set size for $90\%$ reduction of~MSE}
	\label{F:SetSize}
\end{figure*}

\begin{figure}[t]
\centering
\includesvg{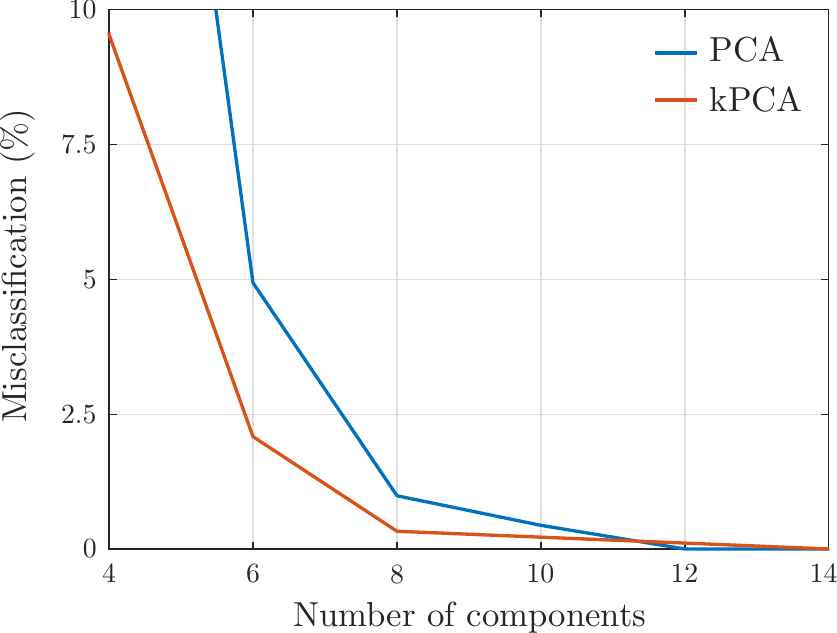}
\caption{Classification performance of PCA and kPCA}
	\label{F:PCAvskPCA}
\end{figure}

In this section, we start by evaluating the performance greedy sampling set selection~(Algorithm~\ref{L:greedyMSE}). For comparison, we also display the results obtained by the \emph{uniform} and \emph{leverage score} randomized methods from~\cite{Chen16s} and a \emph{deterministic} heuristic based on sampling nodes with the highest leverage score~($\norm{\bv_i}_2^2$). In the following examples, we use undirected graphs generated using the \emph{Erd\H{o}s-Rényi} model, in which an edge is placed between two nodes with probability $p = 0.2$; the \emph{preferential attachment} model~\cite{Barabasi99e}, in which nodes are added one at a time and connected to a node already in the graph with probability proportional to its degree; and a \emph{random undirected graph}, obtained by assigning a weight to all possible edges uniformly at random from~$[0,1]$. 

The figure of merit in the following simulations is the \emph{relative suboptimality} from~\eqref{E:relOptimality}. Since it depends on the optimal sampling set which needs to be determined by exhaustive search, we focus on graphs with~$n = 20$ nodes. Bandlimited graph signals are generated by taking~$\bV_\calK$ in~\eqref{E:bandlimited} to be the eigenvectors of the graph adjacency matrix relative to the five eigenvalues with largest magnitude~($\abs{\calK} = 5$). The random vectors~$\bxb$ in~\eqref{E:bandlimited} and~$\bw$ in~\eqref{E:y} are realizations of zero-mean Gaussian random variables with covariance matrices~$\bLambda = \bI$ and~$\bLambda_w = \sigma_w^2 \bI$, where~$\sigma_w^2$ is varied to obtain different SNRs. The transform in~\eqref{E:z} is taken to be the identity~($\bH = \bI$) and the sampling set size is chosen as~$\ell = \abs{\calK} = 5$.

Figures~\ref{F:LowSNR} and~\ref{F:HighSNR} display histograms of the relative suboptimality for $1000$~realizations of graphs and graph signals with~$\sigma_w^2 = 10^2$~($\text{SNR} = -20$~dB) and~$\sigma_w^2 = 10^{-2}$~($\text{SNR} = 20$~dB), respectively. As predicted by Theorem~\ref{T:main}, greedy sampling set selection performs better in low SNR environments, where the optimal sampling set was obtained in more than~$95$\% of the realizations. Nevertheless, even in high SNRs, it found the optimal sampling set almost half of the time. In fact, note that Algorithm~\ref{L:greedyMSE} typically performs much better than the bounds in Theorem~\ref{T:main}~(see details in Fig.~\ref{F:HighSNR}). For comparison, results for greedily optimizing~$\log\det\left[ \bK^\star(\calS) \right]$, a supermodular function, are shown in Figure~\ref{F:logdet}. Although the MSE is $\alpha$-supermodular with~$\alpha < 1$, the relative suboptimality obtained by using both cost functions is comparable.

It is worth noting that, although the deterministic leverage score ranking technique often yields good results, there are advantages to greedy sampling set selection, specially for higher SNR. The randomized sampling schemes, on the other hand, do not perform as well for single problem instances. To be fair, these methods are more appropriate when several sampling sets of the same graph signal are considered. Indeed, the performance measures in~\cite{Chen16s} hold in expectation over sampling realizations.

Evaluating the relative suboptimality for larger graphs is untractable. However, since these sampling set selection techniques build the sampling set sequentially, we can assess their performance in terms of the sampling set size required to obtain a given MSE reduction. Figure~\ref{F:SetSize} displays the distribution of the sampling set size required to achieve a~$90\%$ reduction in the MSE with respect to the empty set. The plots are obtained from~$1000$~graphs and signals realizations with $n = 100$~nodes, $\bV_\calK$ in~\eqref{E:bandlimited} composed the eigenvectors relative to the seven eigenvalues with largest magnitude~($\abs{\calK} = 7$), and $\sigma_w^2 = 10^{-2}$. Although Theorem~\ref{T:main} estimates that Algorithm~\ref{L:greedyMSE} requires sets considerably larger to recover the same near-optimal guarantees as supermodular functions, greedy sampling obtained a sampling set of size exactly~$\abs{\calK}$ in more than~$50\%$ of the realizations. Moreover, as noted in~\cite{Chen16s}, we can now see that leverage score sampling has similar performance to uniform sampling for Erd\H{o}s-Rényi graphs, but gives better results for the preferential attachment model.

\subsection{Application: Subsampled Kernel PCA}
	\label{S:kPCA}

\begin{figure}[t]
\centering
\includesvg{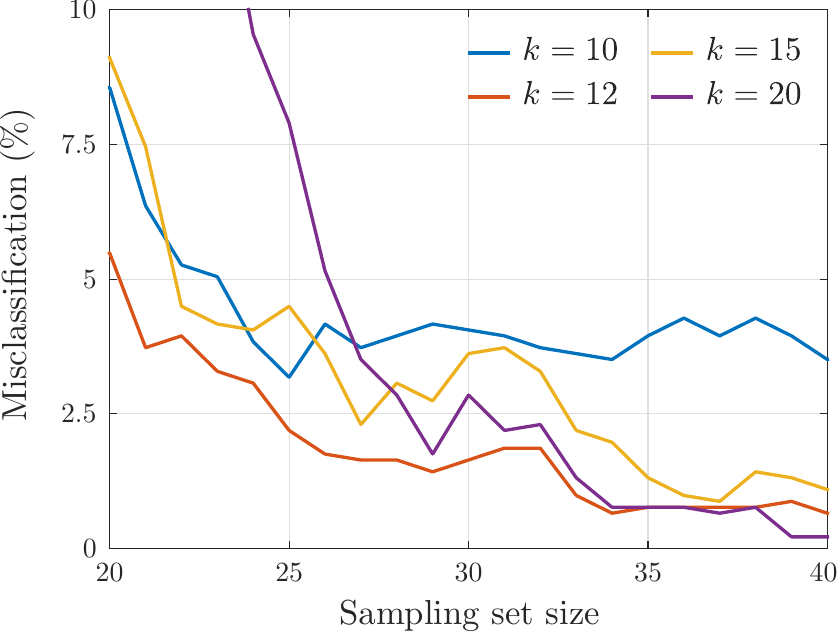}
\caption{Classification performance of greedy subsampled kPCA}
	\label{F:skPCA}
\end{figure}

Kernel PCA is a nonlinear version of PCA~\cite{Scholkopf98n} that also identifies data subspaces by truncating the eigenvalue decomposition~(EVD) of a Gram matrix~$\bm{\Phi}$. However, whereas PCA uses the empirical covariance matrix, kPCA constructs~$\bm{\Phi}$ by evaluating inner products between data points in a higher dimensional space~$\mathbb{F}$ known as the \emph{feature space}. Since the map~$\varphi: \setR^m \to \mathbb{F}$ can be nonlinear and~$\mathbb{F}$ typically has infinite dimensionality, kPCA results in richer subspaces than PCA~\cite{Scholkopf98n, Jeronimo13k, Bishop07p}.

Naturally, the dimensionality of~$\mathbb{F}$ poses a challenge for constructing the Gram matrix. This problem is addressed using the so called \emph{kernel trick}~\cite{Scholkopf98n, Jeronimo13k, Bishop07p}. A kernel is a function~$\kappa$ that allows the inner product in~$\mathbb{F}$ to be evaluated directly from vectors in~$\setR^m$, i.e., $\kappa(\br,\bs) = \langle \varphi(\br), \varphi(\bs) \rangle_\mathbb{F}$. We can use~$\kappa$ to construct~$\bm{\Phi}$ from a training set~$\{\bu_i\}_{i = 1,\dots,n}$, $\bu_i \in \setR^m$, as in
\begin{equation}\label{E:kernelMatrix}
	\bm{\Phi} = \left[ \kappa(\bu_i,\bu_j) \right]_{i,j = 1,\dots,n}
		\text{.}
\end{equation}
Kernel PCA identifies the data subspace as the span of the first~$k$ eigenvectors of~$\bm{\Phi}$, i.e., as~$\colspan(\bV_{\calK})$, where~$\bm{\Phi} = \bV \bLambda \bV^H$ is the EVD of~$\bm{\Phi}$ with eigenvalues in decreasing order and~$\calK = 1,\dots,k$. Using the representer's theorem~\cite{Bishop07p}, any data point~$\by$ can be projected onto this subspace by
\begin{equation}\label{E:kernelProj}
	\byb = \bV_\calK^H \byt
		\text{,} \quad \byt = \left[ \kappa(\bu_i, \by) \right]_{i = 1,\dots,n}
		\text{.}
\end{equation}

\begin{figure}[t]
\centering
\includesvg{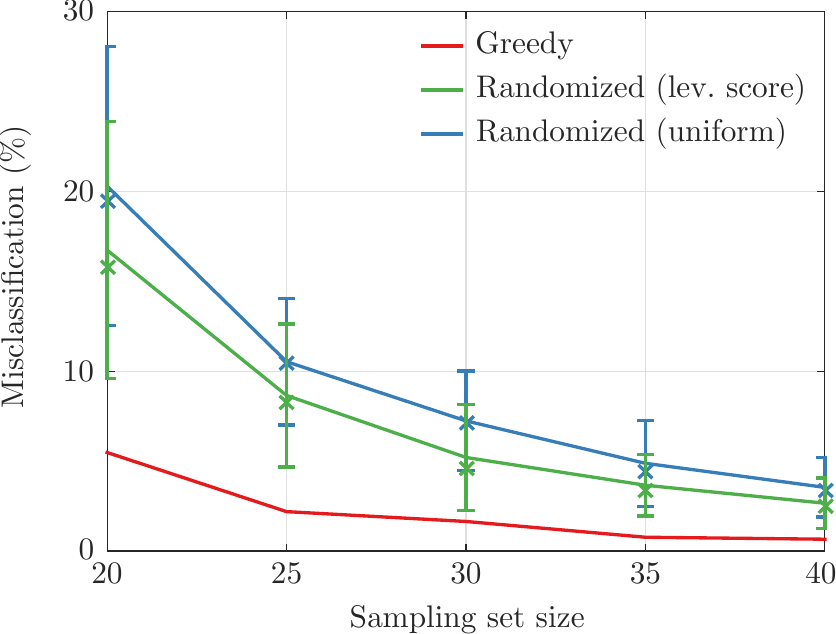}
\caption{Classification performance of subsampled kPCA for different sampling schemes~($k = 12$~components): mean~(line), median~($\times$), and error bars~(one standard deviation) based on~$100$ sampling realizations.}
	\label{F:skPCA_k12}
\end{figure}

The projection in~\eqref{E:kernelProj} requires~$\Theta(kn)$~operations and $n$~KEs, making this method impractical for large data sets even if the dimension~$k$ of the subspace of interest is small. Indeed, although the training phase in~\eqref{E:kernelMatrix} is usually performed offline, \eqref{E:kernelProj} needs to be evaluated during the operation phase for every new data point. In~\cite{Tipping00s}, this issue was addressed by using a Gaussian generative model for~$\bm{\Phi}$ and showing that its maximum likelihood estimate depends only on a subset of the~$\bu_i$. Another approach is to impose sparsity on~$\bV$ \emph{a priori} so that it depends only on a reduced number of training points~\cite{Jeronimo13k}. Alternatively, one can find a representative subset of the training data and apply kPCA to that subset~\cite{Washizawa09s}. The issue with the latter method is that finding a good data subset is known to be a hard problem~\cite{Woodruff14s, Feldman13t}. In fact, it is related to the problem of sampling set selection in GSP.

Indeed, since we used the same notation as in Section~\ref{S:GraphSignal}, formulating kPCA in the context of GSP is straightforward. Let the graph~$\graph$ have adjacency matrix~$\bA = \bm{\Phi}$, which is symmetric and normal, so that~\eqref{E:kernelProj} has the form of a~(partial) graph Fourier transform~\eqref{E:bandlimited}. In other words, \eqref{E:kernelProj} can be interpreted as enforcing graph signals of the form~$\byt$ to be bandlimited on~$\bm{\Phi}$. Thus, we can apply the sampling and interpolation theory from GSP to put forward a \emph{subsampled~kPCA}.

Based on the guarantees given in Section~\ref{S:greedySampling}, we use greedy search to obtain a sampling set~$\calS$ and use the interpolation techniques from Section~\ref{S:Interpolation} to recover~$\byt$ from its samples as in
\begin{equation}\label{E:BLyS}
	\byt = \bL^\star \byt_\calS
		\text{.}
\end{equation}
Then, \eqref{E:kernelProj} and~\eqref{E:BLyS} yield
\begin{equation}\label{E:kernelSampleProj}
	\byb = \underbrace{\bV_\calK^H \bL^\star}_{\bP}
		\byt_\calS
		\text{.}
\end{equation}
Notice that~$\bP$ is now~$k \times \abs{\calS}$, so that the projection in~\eqref{E:kernelSampleProj} only takes~$\Theta(k \abs{\calS})$ operations and~$\abs{\calS}$~KEs, leading to a considerable complexity reduction~($\abs{\calS}/n$) over the direct projection in~\eqref{E:kernelProj}. Moreover, kPCA is typically used for dimensionality reduction prior to regression or classification, so that we are actually interested in a linear transformation of~$\byb$. Subsampled kPCA can account for this case by properly choosing~$\bH$ in~\eqref{E:z}. It is worth noting that contrary to~\cite{Washizawa09s}, the full dataset is used during the training stage to obtain~$\bV_{\calK}$. However, once~$\bP$ is determined, only the subset~$\calS$ is required.

\begin{figure}[t]
\centering
\includesvg{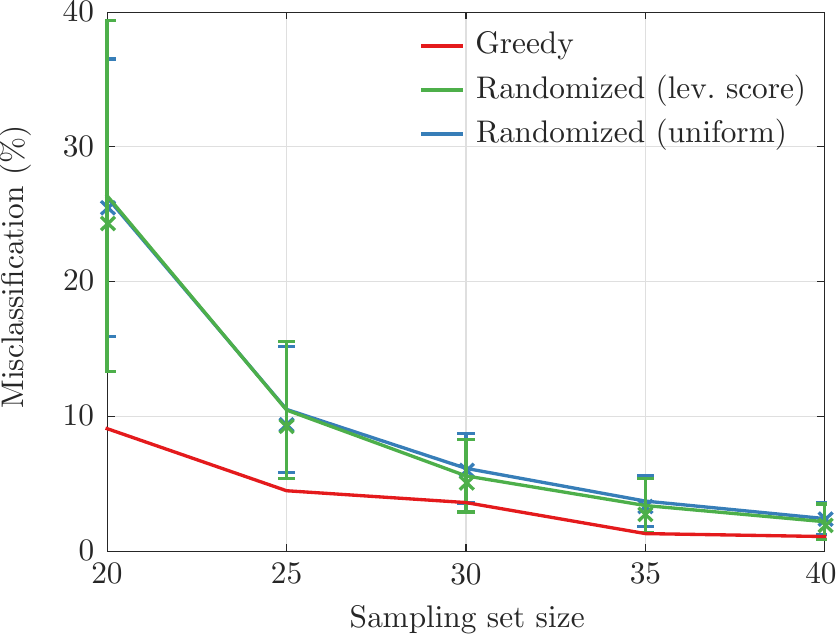}
\caption{Classification performance of subsampled kPCA for different sampling schemes~($k = 15$~components): mean~(line), median~($\times$), and error bars~(one standard deviation) based on~$100$ sampling realizations.}
	\label{F:skPCA_k15}
\end{figure}

In the sequel, we illustrate this method in a face recognition application using the \emph{faces94} data set~\cite{faces94}. It contains~$20$ pictures~($200 \times 180$) of~$152$ individuals which were converted to black and white and normalized so that the value of each pixel is in~$[-1,1]$. A training set is obtained by randomly choosing~$14$ images for each individual~($70$\% of the data set) and the remaining pictures are used for testing. In this application, we use a polynomial kernel of degree~$d = 2$~\cite{Bishop07p} and a one-against-one multiclass support vector machine~(SVM) classifier, in which an SVM is trained for each pair of class and the classification is obtained by majority voting~(see~\cite{Hsu02c} for details on this scheme). Finally, note that since images in both training and testing sets come from the same data set there is no observation noise~$\bw$. Still,~$\sigma_w^2$ can be used to regularize the matrix inversions in~\eqref{E:normalEq} and~\eqref{E:Kstar}~\cite{Kailath00l}.

Figure~\ref{F:PCAvskPCA} shows the misclassification percentage on the test set as a function of the number of components~($k$) for both PCA and kPCA. Note that kPCA can achieve the same performance as PCA with less components. The results of using greedy subsampled kPCA are shown in Figure~\ref{F:skPCA} and clearly illustrate the trade-off between complexity and performance: as the sampling set size increases, the classification errors decrease. However, since misclassification is a nonlinear function of the MSE, it may be advantageous to use more components instead of increasing the sampling set. For instance, kPCA requires~$k = 7$ components to achieve a misclassification of~$1$\%, so that evaluating the direct projection in~\eqref{E:kernelProj} takes $2128$~KEs and $29785$~operations. Greedy subsampled kPCA, on the other hand, can achieve the same performance with~$k = 12$ components and~$\abs{\calS} = 33$, i.e., $33$~KEs and $780$~operations, a complexity reduction of more than~$97$\%. Nevertheless, using $7$~components, greedy subsampled kPCA would require~$\calS$ to be almost the full training set.

Naturally, the method in~\eqref{E:kernelMatrix}--\eqref{E:kernelSampleProj} is not restricted to sampling sets obtained greedily. Thus, in Figures~\ref{F:skPCA_k12} and~\ref{F:skPCA_k15} we compare greedy sampling to the other methods from Section~\ref{S:Simulations} now based on their misclassification performance for~$12$ and~$15$ components. We omit the deterministic leverage score results because it performed consistently worst than the other methods. The average misclassification rates for the randomized schemes are from~$1$ to~$15\%$ higher than those of greedy sampling. Although some realizations yield good classification, their performance varies a lot, especially for smaller sampling sets. Comparing Figures~\ref{F:skPCA_k12} and~\ref{F:skPCA_k15}, it is ready that in this application the difference between uniform and leverage score sampling becomes less significant as the number of components increases.

\begin{remark}

Although this section discussed kPCA, the same argument applies to the classical PCA. It is therefore straightforward to derive an analog \emph{subsampled PCA} technique using~\eqref{E:kernelMatrix}--\eqref{E:kernelSampleProj}.

\end{remark}

\section{Conclusion}

This work provided a solution to graph signal sampling problems by addressing the issue of sampling set selection in two ways. First, it derived universal bounds on the interpolation MSE~(Theorem~\ref{T:mseBound}) which allow the quality of any sampling set or sampling heuristic to be evaluated by gauging how close their reconstruction performance is to the lower bound. Second, it provided near-optimality results for greedy MSE minimization~(Theorem~\ref{T:main}), demonstrating that greedy sampling set selection is an effective sampling scheme, justifying its empirical success in the literature.
The strength of Theorem~\ref{T:main} is that it gives a worst-case result: there exists no graph or graph signal for which the relative suboptimality of greedy sampling is worst than~$e^{-\alpha}$. In fact, greedy sampling typically performs much better and consistently across graph signal realizations. In contrast, although randomized sampling schemes can be effective, their performance can vary widely across realizations.
We should note that given the generality of Theorems~\ref{T:greedy} and~\ref{T:main}, it is likely that they can be applied beyond GSP for sensor placement, experimental design, and variable selection. Moreover, despite the MSE's ubiquity in signal processing, other performance metrics are sometimes more appropriate and we foresee that the theory from this paper can be extended to these cases. In particular, we believe that the concept of approximate submodularity can be used to provide guarantees for the greedy minimization of other non-supermodular functions.

\appendices

\section{Proof of Lemma~\ref{T:MSEFull}}
	\label{A:ThmMseFull}

\begin{proof}

Once again, part~(i) is a corollary of Lemma~\ref{T:monotonicity}. To prove part~(ii), we proceed as for Lemma~\ref{T:MSE}. However, we now let $\bZ(\calA) = \bLambda^{-1} + \sum_{i \in \calA} \lambda_{w,i}^{-1} \bv_i^{} \bv_i^H$ so that again the increment in~\eqref{E:alpha} reads
\begin{multline*}
	f(\calA \cup \{u\}) - f(\calA) ={}
	\\
	\trace \left[
		\bW \left( \bZ(\calA) + \lambda_{w,u}^{-1} \bv_u^{} \bv_u^H \right)^{-1}
		- \bW \bZ(\calA)^{-1}
	\right]
		\text{,}
\end{multline*}
which using the matrix inversion simplifies to
\begin{align*}
	f(\calA \cup \{u\}) - f(\calA) &=
	- \trace \left[
		\bW \frac{
			\bZ(\calA)^{-1} \bv_u \bv_u^H \bZ(\calA)^{-1}
		}{
			\lambda_{w,u}^{-1} + \bv_u^H \bZ(\calA)^{-1} \bv_u
		}
	\right]
	\\
	{}&=
	- \frac{
		\bv_u^H \bZ(\calA)^{-1} \bW \bZ(\calA)^{-1} \bv_u
	}{
		\lambda_{w,u}^{-1} + \bv_u^H \bZ(\calA)^{-1} \bv_u
	}
		\text{.}
\end{align*}
Using this expression, the expression for $\alpha$ in~\eqref{E:alpha} becomes
\begin{equation}\label{E:alphaMSEFull2}
\alpha = \min_{\substack{\calA \subseteq \calB \subseteq \calV \\ u \notin \calB}}
		\frac{
			\lambda_{w,u}^{-1} + \bv_u^H \bZ(\calB)^{-1} \bv_u
		}{
			\lambda_{w,u}^{-1} + \bv_u^H \bZ(\calA)^{-1} \bv_u
		}
		\frac{
			\bv_u^H \bZ(\calA)^{-1} \bW \bZ(\calA)^{-1} \bv_u
		}{
			\bv_u^H \bZ(\calB)^{-1} \bW \bZ(\calB)^{-1} \bv_u
		}
	\text{.}
\end{equation}
Note that, although similar, \eqref{E:alphaMSEFull2} is not the same as~\eqref{E:alpha2}.

We now bound~\eqref{E:alphaMSEFull2} by noticing that for any set~$\calX \subseteq \calV$
\begin{multline*}
	\mu_\text{min} \leq
	\lambda_\text{min} \left[ \bLambda^{-1} \right] \leq
	\lambda_\text{min} \left[ \bZ(\calX) \right] \leq{}
	\\
	\lambda_\text{max} \left[ \bZ(\calX) \right] \leq
	\lambda_\text{max} \left[ \bLambda^{-1} + \bV_\calK^H \bLambda_{w}^{-1} \bV_\calK \right]
	\leq \mu_\text{max}
		\text{.}
\end{multline*}
Thus, using the Rayleigh quotient inequalities leads to
\begin{equation*}
\alpha \geq \frac{
	\lambda_{w,u}^{-1} + \norm{\bv_u}_2^2 \lambda_\text{max} \left[ \bZ(\calB) \right]^{-1}
}{
	\lambda_{w,u}^{-1} + \norm{\bv_u}_2^2 \lambda_\text{min} \left[ \bZ(\calA) \right]^{-1}
}
\frac{
	\lambda_\text{min} \left[ \bZ(\calA)^{-1} \bW \bZ(\calA)^{-1} \right]
}{
	\lambda_\text{max} \left[ \bZ(\calB)^{-1} \bW \bZ(\calB)^{-1} \right]
}
	\text{,}
\end{equation*}
which can be simplified using the same singular value bounds as in Lemma~\ref{T:MSE}~\cite[Thm.~9.H.1, p.~338]{Marshall09i} to yield
\begin{equation}\label{E:alphaMSEFull3}
\alpha \geq \frac{
	\lambda_{w,u}^{-1} + \norm{\bv_u}_2^2 \mu_\text{max}^{-1}
}{
	\lambda_{w,u}^{-1} + \norm{\bv_u}_2^2 \mu_\text{min}^{-1}
}
\frac{
	\mu_\text{max}^{-2}
}{
	\kappa_2(\bW) \, \mu_\text{min}^{-2}
}
\triangleq \alpha^\prime
	\text{,}
\end{equation}
where again~$\kappa_2(\bW) = \lambda_\text{max}(\bW) / \lambda_\text{min}(\bW)$ is the 2-norm condition number of $\bW$. To obtain the expression in~\eqref{E:alphaMSEFull}, notice that~\eqref{E:alphaMSEFull3} is decreasing with respect to~$\norm{\bv_u}_2^2$ and~$\lambda_{w,u}^{-1}$. Indeed,
\begin{align*}
	\frac{\partial \alpha^\prime}{\partial \norm{\bv_u}_2^2} =
	\frac{
		\mu_\text{max}^{-2}
	}{
		\kappa_2(\bW) \, \mu_\text{min}^{-2}
	}
	\frac{
		\lambda_{w,u}^{-1}
		\left( \mu_\text{max}^{-1} - \mu_\text{min}^{-1} \right)
	}{
		\left( \lambda_{w,u}^{-1} +
			\norm{\bv_u}_2^2 \mu_\text{min}^{-1} \right)^2
	} \leq 0
	\\
	\frac{\partial \alpha^\prime}{\partial \lambda_{w,u}} =
	\frac{
		\mu_\text{max}^{-2}
	}{
		\kappa_2(\bW) \, \mu_\text{min}^{-2}
	}
	\frac{
		\lambda_{w,u}^{-2} \norm{\bv_u}_2^2
		\left( \mu_\text{max}^{-1} - \mu_\text{min}^{-1} \right)
	}{
		\left( \lambda_{w,u}^{-1} +
			\norm{\bv_u}_2^2 \mu_\text{min}^{-1} \right)^2
	} \leq 0
\end{align*}
are both non-positive because~$0 < \mu_\text{min} \leq \mu_\text{max}$ and~$\kappa_2(\bW) \geq 1$~\cite{Horn13}. We then use the fact that~$\norm{\bv_u}_2^2 \leq 1$ to get~\eqref{E:alphaMSEFull}.
\end{proof}

\IEEEtriggeratref{27}
\bibliographystyle{IEEEtran}
\bibliography{IEEEabrv,gsp,sp,math}

\begin{thebibliography}{10}
\providecommand{\url}[1]{#1}
\csname url@samestyle\endcsname
\providecommand{\newblock}{\relax}
\providecommand{\bibinfo}[2]{#2}
\providecommand{\BIBentrySTDinterwordspacing}{\spaceskip=0pt\relax}
\providecommand{\BIBentryALTinterwordstretchfactor}{4}
\providecommand{\BIBentryALTinterwordspacing}{\spaceskip=\fontdimen2\font plus
\BIBentryALTinterwordstretchfactor\fontdimen3\font minus
  \fontdimen4\font\relax}
\providecommand{\BIBforeignlanguage}[2]{{%
\expandafter\ifx\csname l@#1\endcsname\relax
\typeout{** WARNING: IEEEtran.bst: No hyphenation pattern has been}%
\typeout{** loaded for the language `#1'. Using the pattern for}%
\typeout{** the default language instead.}%
\else
\language=\csname l@#1\endcsname
\fi
#2}}
\providecommand{\BIBdecl}{\relax}
\BIBdecl

\bibitem{Shuman13e}
D.~Shuman, S.~Narang, P.~Frossard, A.~Ortega, and P.~Vandergheynst, ``The
  emerging field of signal processing on graphs: Extending high-dimensional
  data analysis to networks and other irregular domains,'' \emph{{IEEE} Signal
  Process. Mag.}, vol. 30[3], pp. 83--98, 2013.

\bibitem{Sandryhaila13d}
A.~Sandryhaila and J.~Moura, ``Discrete signal processing on graphs,''
  \emph{{IEEE} Trans. Signal Process.}, vol. 61[7], pp. 1644--1656, 2013.

\bibitem{Narang12p}
S.~Narang and A.~Ortega, ``Perfect reconstruction two-channel wavelet filter
  banks for graph structured data,'' \emph{{IEEE} Trans. Signal Process.}, vol.
  60[6], pp. 2786--2799, 2012.

\bibitem{Tremblay16c}
N.~Tremblay, G.~Puy, R.~Gribonval, and P.~Vandergheynst, ``Compressive spectral
  clustering,'' in \emph{Int. Conf. on Mach. Learning}, 2016, pp.
  1002–--1011.

\bibitem{Zhu12a}
X.~Zhu and M.~Rabbat, ``Approximating signals supported on graphs,'' in
  \emph{Int. Conf. on Acoust., Speech and Signal Process.}, 2012, pp.
  3921--3924.

\bibitem{Pesenson08s}
I.~Pesenson, ``Sampling in paley-wiener spaces on combinatorial graphs,''
  \emph{Trans. of the American Mathematical Society}, vol. 360[10], pp.
  5603--5627, 2008.

\bibitem{Pesenson10s}
I.~Pesenson and M.~Pesenson, ``Sampling, filtering and sparse approximations on
  combinatorial graphs,'' \emph{J. of Fourier Analysis and Applications}, vol.
  16[6], pp. 921--942, 2010.

\bibitem{Narang13s}
S.~Narang, A.~Gadde, and A.~Ortega, ``Signal processing techniques for
  interpolation in graph structured data,'' in \emph{Int. Conf. on Acoust.,
  Speech and Signal Process.}, 2013, pp. 5445--5449.

\bibitem{Shomorony14s}
H.~Shomorony and A.~Avestimehr, ``Sampling large data on graphs,'' in
  \emph{Global Conf. on Signal and Inform. Process.}, 2014, pp. 933--936.

\bibitem{Chen15d}
S.~Chen, R.~Varma, A.~Sandryhaila, and J.~Kovačević, ``Discrete signal
  processing on graphs: Sampling theory,'' \emph{{IEEE} Trans. Signal
  Process.}, vol. 63[24], pp. 6510--6523, 2015.

\bibitem{Anis16e}
A.~Anis, A.~Gadde, and A.~Ortega, ``Efficient sampling set selection for
  bandlimited graph signals using graph spectral proxies,'' \emph{{IEEE} Trans.
  Signal Process.}, vol. 64[14], pp. 3775--3789, 2016.

\bibitem{Tsitsvero16s}
M.~Tsitsvero, S.~Barbarossa, and P.~{Di Lorenzo}, ``Signals on graphs:
  Uncertainty principle and sampling,'' \emph{{IEEE} Trans. Signal Process.},
  vol. 64[18], pp. 4845--4860, 2016.

\bibitem{Chen16s}
S.~Chen, R.~Varma, A.~Singh, and J.~Kovačević, ``Signal recovery on graphs:
  Fundamental limits of sampling strategies,'' \emph{{IEEE} Trans. Signal
  Process.}, vol. 2[4], pp. 539--554, 2016.

\bibitem{Antonio16s}
A.~Marques, S.~Segarra, G.~Leus, and A.~Ribeiro, ``Sampling of graph signals
  with successive local aggregations,'' \emph{{IEEE} Trans. Signal Process.},
  vol. 64[7], pp. 1832--1843, 2016.

\bibitem{Chepuri16s}
S.~Chepuri and G.~Leus, ``Subsampling for graph power spectrum estimation,'' in
  \emph{Sensor Array and Multichannel Signal Process. Workshop}, 2016.

\bibitem{Krause08n}
A.~Krause, A.~Singh, and C.~Guestrin, ``Near-optimal sensor placements in
  {G}aussian processes: Theory, efficient algorithms and empirical studies,''
  \emph{J. Mach. Learning Research}, vol.~9, pp. 235--284, 2008.

\bibitem{Das11s}
A.~Das and D.~Kempe, ``Submodular meets spectral: Greedy algorithms for subset
  selection, sparse approximation and dictionary selection,'' in \emph{Int.
  Conf. on Mach. Learning}, 2011.

\bibitem{Sagnol13a}
G.~Sagnol, ``Approximation of a maximum-submodular-coverage problem involving
  spectral functions, with application to experimental designs,''
  \emph{Discrete Appl. Math.}, vol. 161[1-2], pp. 258--276, 2013.

\bibitem{Ranieri14n}
J.~Ranieri, A.~Chebira, and M.~Vetterli, ``Near-optimal sensor placement for
  linear inverse problems,'' \emph{{IEEE} Trans. Signal Process.}, vol. 62[5],
  pp. 1135--1146, 2014.

\bibitem{Gama16r}
F.~Gama, A.~Marques, G.~Mateos, and A.~Ribeiro, ``Rethinking sketching as
  sampling: {L}inear transforms of graph signals,'' in \emph{Asilomar Conf. on
  Signals, Syst. and Comput.}, 2016.

\bibitem{Thanou14l}
D.~Thanou, D.~Shuman, and P.~Frossard, ``Learning parametric dictionaries for
  signals on graphs,'' \emph{{IEEE} Trans. Signal Process.}, vol. 62[15], pp.
  3849--3862, 2014.

\bibitem{Mirzasoleiman15l}
B.~Mirzasoleiman, A.~Badanidiyuru, A.~Karbasi, J.~Vondrák, and A.~Krause,
  ``Lazier than lazy greedy,'' in \emph{AAAI Conf. on Artificial Intell.},
  2015, pp. 1812--1818.

\bibitem{Bach14l}
F.~Bach, ``Learning with submodular functions: A convex optimization
  perspective,'' \emph{Foundations and Trends in Machine Learning}, vol.
  6[2-3], pp. 145--373, 2013.

\bibitem{Chamon16n}
L.~Chamon and A.~Ribeiro, ``Near-optimality of greedy set selection in the
  sampling of graph signals,'' in \emph{Global Conf. on Signal and Inform.
  Process.}, 2016.

\bibitem{Woodruff14s}
D.~Woodruff, ``Sketching as a tool for numerical linear algebra,''
  \emph{Foundations and Trends in Theoretical Computer Science}, vol. 10[1-2],
  pp. 1--157, 2014.

\bibitem{Feldman13t}
D.~Feldman, M.~Schmidt, and C.~Sohler, ``Turning big data into tiny data:
  Constant-size coresets for {K-means}, {PCA} and projective clustering,'' in
  \emph{ACM-SIAM Symp. on Discrete Algorithms}, 2013, pp. 1434--1453.

\bibitem{Scholkopf98n}
B.~Schölkopf, , A.~Smola, and K.-R. Müller, ``Nonlinear component analysis as
  a kernel eigenvalue problem,'' \emph{Neural Comput.}, vol. 10[5], pp.
  1299--1319, 1998.

\bibitem{Jeronimo13k}
J.~{Arenas-Garcia}, K.~Petersen, G.~{Camps-Valls}, and L.~Hansen, ``Kernel
  multivariate analysis framework for supervised subspace learning: A tutorial
  on linear and kernel multivariate methods,'' \emph{{IEEE} Signal Process.
  Mag.}, vol. 30[4], pp. 16--29, 2013.

\bibitem{Kailath00l}
T.~Kailath, A.~Sayed, and B.~Hassibi, \emph{Linear estimation}.\hskip 1em plus
  0.5em minus 0.4em\relax Prentice-Hall, 2000.

\bibitem{Horn13}
R.~Horn and C.~Johnson, \emph{Matrix analysis}.\hskip 1em plus 0.5em minus
  0.4em\relax Cambridge University Press, 2013.

\bibitem{Adali14o}
T.~Adali and P.~Schreier, ``Optimization and estimation of complex-valued
  signals: {T}heory and applications in filtering and blind source
  separation,'' \emph{{IEEE} Signal Process. Mag.}, vol. 31[5], pp. 112--128,
  2014.

\bibitem{Chamon16u}
L.~Chamon and A.~Ribeiro, ``Universal bounds for the sampling of graph
  signals,'' in \emph{Int. Conf. on Acoust., Speech and Signal Process.}, 2017.

\bibitem{Girault15s}
B.~Girault, ``Stationary graph signals using an isometric graph translation,''
  in \emph{European Signal Process. Conf.}, 2015, pp. 1516--1520.

\bibitem{Antonio16st}
A.~Marques, S.~Segarra, G.~Leus, and A.~Ribeiro, ``Stationary graph processes
  and spectral estimation,'' 2016, arXiv:1603.04667v1.

\bibitem{Perraudin16s}
N.~Perraudin and P.~Vandergheynst, ``Stationary signal processing on graphs,''
  \emph{{IEEE} Trans. Signal Process.}, vol. 65[13], pp. 3462--3477, 2017.

\bibitem{Lorenzo16a}
P.~{Di Lorenzo}, S.~Barbarossa, P.~Banelli, and S.~Sardellitti, ``Adaptive
  least mean squares estimation of graph signals,'' \emph{{IEEE} Trans. Signal
  Inf. Process. over Netw.}, vol. 2[4], pp. 555--568, 2016.

\bibitem{Wang15l}
X.~Wang, P.~Liu, and Y.~Gu, ``Local-set-based graph signal reconstruction,''
  \emph{{IEEE} Trans. Signal Process.}, vol. 63[9], pp. 2432--2444, 2015.

\bibitem{Boyd04c}
S.~Boyd and L.~Vandenberghe, \emph{Convex optimization}.\hskip 1em plus 0.5em
  minus 0.4em\relax Cambridge, 2004.

\bibitem{Bhatia97m}
R.~Bhatia, \emph{Matrix analysis}.\hskip 1em plus 0.5em minus 0.4em\relax
  Springer, 1997.

\bibitem{Natarajan95s}
B.~Natarajan, ``Sparse approximate solutions to linear systems,'' \emph{SIAM
  Journal on Computing}, vol. 24[2], pp. 227--234, 1995.

\bibitem{Nemhauser78a}
G.~Nemhauser, L.~Wolsey, and M.~Fisher, ``An analysis of approximations for
  maximizing submodular set functions---{I},'' \emph{Mathematical Programming},
  vol. 14[1], pp. 265--294, 1978.

\bibitem{Joshi09s}
S.~Joshi and S.~Boyd, ``Sensor selection via convex optimization,''
  \emph{{IEEE} Trans. Signal Process.}, vol. 57[2], pp. 451--462, 2009.

\bibitem{Marshall09i}
A.~Marshall, I.~Olkin, and B.~Arnold, \emph{Inequalities: {T}heory of
  Majorization and Its Applications}.\hskip 1em plus 0.5em minus 0.4em\relax
  Springer, 2009.

\bibitem{Barabasi99e}
A.-L. Barabási and R.~Albert, ``Emergence of scaling in random networks,''
  \emph{Science}, vol. 286[5439], pp. 509--512, 1999.

\bibitem{Bishop07p}
C.~Bishop, \emph{Pattern recognition and machine learning}.\hskip 1em plus
  0.5em minus 0.4em\relax Springer, 2007.

\bibitem{Tipping00s}
M.~Tipping, ``Sparse kernel principal component analysis,'' in \emph{Conf. on
  Neural Inform. Process. Syst.}, 2000.

\bibitem{Washizawa09s}
Y.~Washizawa, ``Subset kernel principal component analysis,'' in \emph{Int.
  Workshop on Mach. Learning for Signal Process.}, 2009.

\bibitem{faces94}
L.~Spacek, ``Collection of facial images,''
  \url{http://cswww.essex.ac.uk/mv/allfaces/index.html}.

\bibitem{Hsu02c}
C.-W. Hsu and C.-J. Lin, ``A comparison of methods for multiclass support
  vector machines,'' \emph{{IEEE} Trans. Neural Netw.}, vol. 13[2], pp.
  415--425, 2002.

\end{thebibliography}

\end{document}